\documentclass{article}
\usepackage[a4paper, margin=1.2in]{geometry}
\usepackage{authblk}

\usepackage{url}
\usepackage[hidelinks]{hyperref}
\usepackage[utf8]{inputenc}
\usepackage[small]{caption}
\usepackage{graphicx}
\usepackage{amsmath}
\usepackage{amsthm}
\usepackage{booktabs}
\usepackage{algorithm}
\usepackage{algorithmic}
\usepackage{amssymb,amsfonts}
\usepackage{dsfont}

\usepackage[T1]{fontenc}
\newcommand{\citet}[1]{\citeauthor{#1}~[\citeyear{#1}]}
\usepackage{tikz}

\usepackage[backend=biber,maxbibnames = 100]{biblatex}
\DeclareNameAlias{default}{family-given}

\usepackage{amsthm}
\usepackage{amsmath}
\usepackage{amsfonts}
\usepackage{tikz}
\usepackage[T1]{fontenc}
\usepackage{subcaption}

\newtheorem{example}{Example}
\newtheorem{theorem}{Theorem}
\newtheorem{lemma}[theorem]{Lemma}
\newtheorem{proposition}[theorem]{Proposition}

\newcommand{\e}{\! : \!}
\newcommand{\ee}{\!\! : \!\!}

\newcommand{\appname}{} 
\newcommand{\appnumber}{} 
\newtheorem*{appp}{\appname\ \appnumber}
\newenvironment{apptheorem}[2]
  {\renewcommand{\appname}{#1}%
   \renewcommand{\appnumber}{#2}%
   \begin{appp}}
  {\end{appp}}

\title{PageRank for Edges: Axiomatic Characterization}

\author{Natalia Kucharczuk}

\author{Tomasz Wąs}

\author{Oskar Skibski%
\thanks{\texttt{nk406686@students.mimuw.edu.pl, t.was@mimuw.edu.pl, o.skibski@mimuw.edu.pl} \\
This work is an extended version of \citet{Kucharczuk:etal:2022:edgepagerank} that will appear in the Proceedings of the 36th AAAI Conference on Artificial Intelligence (AAAI-22).
It includes the complete proofs of all theorems which are omitted in the conference publication. \\
This work was supported by the Polish National Science Centre grant 2018/31/B/ST6/03201.
Natalia Kucharczuk was supported by the Ministry of Science and Higher Education project Szkola Orlow, project number 500-D110-06-0465160.}
}

\affil{University of Warsaw}

\date{}

\begin{document}

\maketitle

\begin{abstract}
Edge centrality measures are functions that evaluate the importance of edges in a network.
They can be used to assess the role of a backlink for the popularity of a website as well as the importance of a flight in virus spreading.
Various node centralities have been translated to apply for edges, including Edge Betweenness, Eigenedge (edge version of Eigenvector centrality) and Edge PageRank.
With this paper, we initiate the discussion on the axiomatic properties of edge centrality measures.
We do it by proposing an axiomatic characterization of Edge PageRank.
Our characterization is the first characterization of any edge centrality measure in the literature.
\end{abstract}

\section{Introduction}

Centrality measures that evaluate the importance of nodes and edges in a network constitute one of the fundamental tools of network analysis~\cite{Brandes:Erlebach:2005,Jackson:2005}. 
In complex networks that describe the surrounding world, they enable us to indicate the most significant genes~\cite{Oezguer:etal:2008}, key terrorists~\cite{Krebs:2002:mapping} and important pages in the World Wide Web~\cite{Page:etal:1999}.

Historically, centrality analysis was developed in social networks literature.
Hence, the vast majority of work concentrates on nodes which represent people in such networks.
However, in many types of networks, edges represent entities that we want to assess.
In particular, edge evaluation may indicate which backlink to our website is the most profitable in Search Engine Optimization~\cite{Ledford:2015} or which flight should be canceled in order to delay virus spreading~\cite{Marcelino:Kaiser:2012}.
Edge centralities have also been used to identify edges that connect different communities in clustering algorithms~\cite{Newman:2004}.

Various node centralities have been translated to apply for edges, including \emph{Edge Betweenness}~\cite{Girvan:Newman:2002}, \emph{Eigenedge} (edge version of Eigenvector centrality)~\cite{Huang:Huang:2019}, and \emph{Edge PageRank}~\cite{Chapela:etal:2015}. 
However, other measures not related to any node concepts have also been proposed in the literature, such as \emph{Spanning Edge Betweenness} defined as the fraction of spanning trees that contain a specific edge~\cite{Teixeira:etal:2013}.

Multiplicity of centrality measures constitute a problem of its own, as it is becoming harder to choose one measure for a specific application.
In result, more than often a measure to use is selected based on its intuitive understanding or popularity rather than its suitability and desirable features.
That is why, in recent years, efforts at organizing the space of centrality measures intensified~\cite{Schoch:Brandes:2015,Bloch:etal:2016}.

The axiomatic approach is one of the most convenient methods for this goal as it highlights similarities and differences between various concepts.
In this approach, simple and desirable properties are identified that capture specific features of centrality measures.
Choosing a carefully designed set of axioms allows to create a unique characterization of the measure that is more intuitive and easier to relate to the considered application.

To date, plenty of node centrality measures have been axiomatized~\cite{Boldi:Vigna:2014}.
Notably, much research focused on feedback centralities~\cite{Altman:Tennenholtz:2005,Dequiedt:Zenou:2014}, although distance-based centralities~\cite{Garg:2009} and game-theoretic centralities~\cite{Skibski:etal:2019:attachment} have also received considerable attention.
However, to the best of our knowledge, so far no paper has considered edge centrality measures.

With this paper, we initiate the discussion on the axiomatic properties of edge centrality measures.
We do it by proposing an axiomatic characterization of Edge PageRank.
Edge PageRank, analogically to standard PageRank, can be defined as the unique solution to the system of recursive equations that ties centralities of incident edges.
Our characterization is the first characterization of any edge centrality.

We base our work on a recent characterization of (node) PageRank~\cite{Was:Skibski:2020:pagerank} and ask: is it possible to adapt such a characterization of the node centrality for the edge centrality measure? 
We answer positively to this question.
Specifically, we define six axioms that correspond to axioms proposed for PageRank and show that Edge PageRank is the only measure that satisfies all of them.


The main technical challenge comes from the fact that edge centrality measures consider pairs of nodes, which adds a new dimension to the complexity of the problem.
In result, analyzing edge centralities cannot be reduced to the analysis of node centralities.
We discuss this on the example of line graph approach.
Moreover, this means that also the expressive power of axioms can vary.
For these reasons, our proof is essentially different from the proof for (node) PageRank.

\section{Preliminaries}
In this paper, we consider directed multigraphs with node weights and possible self-loops.
Since the emphasis of our work is on edges, not nodes, we need a model of a multigraph in which every edge is a separate entity. 
To this end, in our model each edge has its own label and an additional function specifies the start and the end of this edge.
See Figure~\ref{fig:small_graph} for an illustration.
In this way, we do not assume that all edges between the same pair of nodes are equally important which could lead to other undesirable (or unexpected) implications.\footnote{To give an example, if we consider unlabeled edges and know that changing the end of edge $e_5$ from $v_2$ to $v_1$ in the graph from Figure~\ref{fig:small_graph} does not affect centrality of any edge, then we get that edges $e_4$, $e_5$ and $e_6$ all have equal centralities.}

Formally, we define a \textit{graph} as a tuple $G = (V, E, \phi, b)$, where $V$ is a set of \textit{nodes}, $E$ is a set of \textit{edges}, $\phi: E \rightarrow V \times V$ is an \textit{incidence function} mapping every edge to an ordered pair of nodes, and $b: V \rightarrow \mathbb{R}_{\geq 0}$ is a node weight function. 

For an edge $e \in E$ with $\phi(e) = (u,v)$, we denote the start $u$ and the end $v$ by $\phi_1(e)$ and $\phi_2(e)$, respectively. 
This edge is an \emph{outgoing edge} for $u$ and \emph{incoming edge} for $v$.

For node $v$, the set of all incoming (outgoing) edges is denoted by $E^-_v(G)$ ($E^+_v(G)$).
The sizes of these sets are called \emph{in-degree} and \emph{out-degree} of node $v$, i.e., $\deg^-_v(G) = |E^-_v(G)|$ and $\deg^+_v(G) = |E^+_v(G)|$.
The set of all incident edges, both incoming and outgoing, is denoted by $E^{\pm}_v(G)$.

Two nodes $u$ and $v$ are called \textit{out-twins} if there exists a bijection $\psi: E^+_u(G) \rightarrow E^+_v(G)$ such that $\phi_2(e) = \phi_2(\psi (e))$ for every $e \in E^+_u(G)$.

A node $v$ is a \emph{sink} if it has no outgoing edges: $E^+_v(G) = \emptyset$.
If it does not have incoming edges as well, i.e., $E^{\pm}_v(G) = \emptyset$, then it is an \emph{isolated node}.
If $v$ has no incoming, but one outgoing edge, i.e., $E^{\pm}_v(G) = \{e\}$, then this edge is called a \emph{source edge}.
If a source edge $e$ is the only edge incident to its end, i.e., $\phi(e) = (v,u)$ and $E^{\pm}_u(G) = \{e\}$, then we say it is an \emph{isolated edge}.

A \emph{path} from $u$ to $v$ is a sequence of edges $e_1, \ldots, e_k$ such that $\phi_1(e_1) = u, \ \phi_2(e_k) = v$ and $\phi_2(e_{i})=\phi_1(e_{i+1})$ for every $i \in \{1,\ldots,k-1\}$. 
If there exists a path from $u$ to $v$, then $v$ is called a \emph{successor} of $u$.
The set of all successors of $u$ is denoted by $S_u(G)$.

Consider an arbitrary function $f: A \rightarrow X$.
We will use the following shorthand function notation.
For $a \in A$ and $x \in X$ we denote by $f[a \rightarrow x]$ the function obtained by replacing the value of $a$ with $x$:
\[
f[a\rightarrow x](c) = 
\begin{cases}
    x & \textrm{if } a=c,\\
    f(c) & \textrm{otherwise.}
\end{cases}
\]
Also, for $B \subseteq A$, we denote the function $f$ restricted to $B$ by $f|_{B}: B \rightarrow X$.
For a second function $f': A' \rightarrow X$ with $A \cap A' = \emptyset$, we define function $(f+f'): A \cup A' \rightarrow X$ as follows: $(f+f')(c) = f(c)$ for $c \in A$ and $(f+f')(c) = f'(c)$ for $c \in A'$.

The \emph{sum} of two disjoint graphs $G = (V, E, \phi,b)$ and $G' = (V', E', \phi',b')$ with $V \cap V' = \emptyset$ and $E \cap E' = \emptyset$ is defined as $G + G' = (V \cup V', E\cup E', \phi+\phi', b+b')$.

When graph $G=(V,E,\phi,b)$ is known from the context, we will simply write
``$e \e (u,v) \in E$'' by which we understand ``$e \in E$ s.t. $\phi(e) = (u,v)$''.
Also, to denote small graphs we will write:
\[ G = (\{v_1, \dots, v_n\}, \{e_1 : c_1, \dots, e_m : c_m\}, [b_1, \dots, b_n]) \]
which means $G = (\{v_1, \dots, v_n\},\{e_1, \dots, e_m\},\phi, b)$ such that $\phi(e_i) = c_i$ for every $i \in \{1,\dots,m\}$ and $b(v_j) = b_j$ for every $j=\{1,\dots, n\}$.

For a graph $G = (V,E,\phi,b)$ and nodes $u,v \in V$, a graph obtained from \emph{redirecting} node $u$ into node $v$ is denoted by $R_{u\rightarrow v} (G)$ and defined as follows: 
\[ R_{u\rightarrow v} (G) = (V\setminus \{u\}, E\setminus E^+_u(G), \phi', b'), \]
where $b' = (b[v \rightarrow b(u) + b(v)])|_{V \setminus \{u\}}$ and $\phi'(e) = (w,v)$ if $\phi(e) = (w,u)$ and $\phi'(e) = \phi(e)$ otherwise. 



\begin{figure}[t]
\centering
\begin{tikzpicture}[x=2cm,y=2cm] 
  \tikzset{     
    e4c node/.style={circle,draw,minimum size=0.6cm,inner sep=0}, 
    e4c edge/.style={above,font=\footnotesize}
  }
  
  \node[e4c node] (1) at (0.0, 4.1) {$v_1$}; 
  \node[e4c node] (2) at (1.0, 4.1) {$v_2$}; 
  \node[e4c node] (3) at (1.0, 3.3) {$v_3$}; 
  \node[e4c node] (4) at (0.0, 3.3) {$v_4$};

  \path[->,draw,thick]
  (1) edge[e4c edge, out=145, in = 215, looseness=6] node[left] {$e_1$}  (1)
  (1) edge[e4c edge] node[left] {$e_2$}  (4)
  (2) edge[e4c edge] node[above, pos = 0.3] {$e_3$}  (4)
  (3) edge[e4c edge] node[above, pos =0.7] {$e_4$}  (1)
  (3) edge[e4c edge, bend left=20] node[pos=0.3, label={[label distance=-0.2cm]180:$e_5$}] {}  (2)
  (3) edge[e4c edge, bend right=20] node[pos=0.3, label={[label distance=-0.2cm]0:$e_6$}] {}  (2)
  (3) edge[e4c edge, bend left=20] node[label={[label distance=-0.35cm]$e_7$}] {}  (4)
  (4) edge[e4c edge, bend left=20] node[label={[label distance=-0.35cm]$e_8$}] {}  (3)
  ;
\end{tikzpicture}


\caption{An example graph $G=(V,E,\phi,b)$ with nodes $v_1, \dots, v_4$ and edges $e_1, \dots, e_8$.
The incident function, $\phi$, specifies the start and the end of each edge, e.g., $\phi(e_5)=\phi(e_6)=(v_3, v_2)$.
We assume uniform node weights, i.e., $b(v_i) = 1$ for every $v_i \in V$.}
\label{fig:small_graph}
\end{figure}
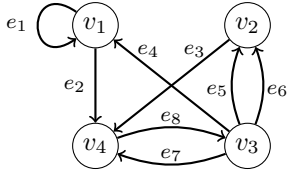










\subsection{Edge PageRank}

An \emph{(edge) centrality measure} $F$ is a function that assesses the importance of an edge $e$ in a graph $G$; this value is denoted by $F_e(G)$ and is non-negative.

\emph{Edge PageRank}~\cite{Chapela:etal:2015} is defined as a unique centrality measure that for every graph $G=(V,E,\phi, b)$ and edge $e \e (u,v) \in E$ satisfies recursive equation
\begin{equation}
\label{eq:pagerank}
    PR^a_{e}(G) = \frac{1}{\deg_{u}^+(G)} \left( a\cdot\sum_{e' \in E^-_u(G)} PR^a_{e'}(G) + b(u)\right),
\end{equation}
where constant $a \in [0,1)$ is a \emph{decay factor}, i.e., a parameter of Edge PageRank.

Equivalently, we can define Edge PageRank using random walks on a graph.
To this end, imagine a surfer that travels throughout a graph in a schematic manner.
She starts her walk from a random node 
(with the probability that the surfer starts in a node proportional to the weight of this node).
Then, in each step, she makes two choices:
first, she chooses one of the outgoing edges of a node she currently occupies, uniformly at random, and follows it to the next node;
second, she decides whether she wants to continue the walk (with probability $a$) or end it (with probability $1-a$).
If the surfer ever arrives at a sink, she ends her walk automatically.
Now, in such a walk, PageRank of an edge is the expected number of times the surfer traversed this edge, multiplied by the sum of node weights in the graph.
We note that this model is a slight variation on a standard random walk model used for interpretation of PageRank.
The only difference between our model and the model presented in~\cite{Was:Skibski:2020:pagerank} is the order of the decisions made by the surfer:
there, she first decides whether she continues the walk and only if so, she chooses an outgoing edge.
In result, the expected number of traverses over an edge is multiplied by $a$ in comparison to our model.

Additionally, the definition of Edge PageRank can be also based on its relation with (node) PageRank.
Specifically, for every graph, $G=(V,E,\phi, b)$, and edge, $e \e (u,v) \in E$, it holds that
\[
    PR^a_e(G) = PR^a_u(G)/\deg^+_u(G),
\]
where $PR^a_{u}(G)$ is PageRank of node $u$ in graph $G$.
This relation is true for multigraphs, which we consider in this paper.
If, instead, we considered graphs with edge weights, $\omega : E \rightarrow \mathbb{R}_{>0}$, then the relation would be
\[
    PR^a_e(G) = \frac{\omega(e)}{\sum_{e' \in E^+_u} \omega(e')} \cdot PR^a_u(G).
\]

Finally, Edge PageRank can be also defined as a (node) PageRank of the corresponding line graph.
We will discuss this in details later in the text.

\begin{example}
An example application for Edge PageRank is to create a ranking of the hyperlinks pointing toward a webpage in the order of their importance for the Internet traffic \cite{Criado:etal:2018}.
Often, the obtained ranking does not coincide with the PageRank ranking of the webpages they are coming from.
To see why, let us consider Edge PageRank of incoming edges of node $v_4$ in graph $G$ from Figure~\ref{fig:small_graph}.
Each of them comes from a different node: $e_2$ from $v_1$, $e_3$ from $v_2$, and $e_7$ from $v_3$.
Among these nodes, the greatest PageRank has $v_3$, followed by $v_1$, and then $v_2$.
However, $e_3$ is the only outgoing edge of node $v_2$, whereas both $v_3$ and $v_1$ have multiple outgoing edges.
In result, Edge PageRank of edge $e_3$ is greater than that of $e_2$ and $e_7$, making it the most important incoming edge of $v_4$.
PageRank values of nodes and edges considered in this example for $a=0.9$ are presented in the following table:

\begin{table}[h]
\centering
{\setlength{\tabcolsep}{3pt}
\renewcommand{\arraystretch}{1.2}
\textit{
\begin{tabular}{c|c}
    node $v$  & $PR^{0.9}_v(G,b)$ \\
    \hline
    $v_1$ &  7.09 \\
    $v_2$ &  6.80 \\
    $v_3$ &  12.89 
\end{tabular}
\quad
\begin{tabular}{c|c|c}
    edge $e$  & $\phi(e)$ & $PR^{0.9}_e(G,b)$ \\
    \hline
    $e_2$ & $(v_1,v_4)$ & 3.55 \\
    $e_3$ & $(v_2,v_4)$ & 6.80 \\
    $e_7$ & $(v_3,v_4)$ & 3.22 
\end{tabular}
}
}
\end{table}
\end{example}


\section{Axioms}

In this section, we propose our axioms that are based on the axioms for (node) PageRank introduced in \cite{Was:Skibski:2020:pagerank}.
While some of our axioms are straightforward adaptations of the original axioms, some required significant modifications in order to work for an edge centrality.
What is important, is that these six adapted axioms now uniquely characterize Edge PageRank (what we prove in Theorem~\ref{theorem:main}).

The axioms are as follows:
\begin{itemize}
    \item \textbf{Node Deletion:} 
    For every graph $G = (V, E, \phi, b)$ and isolated node $u \in V$ it holds that
    \[
        F_e(V\setminus \{u\}, E, \phi, b|_{V\setminus\{u\}}) = F_e(G)
    \]
    for every $e \in E$.
    \item \textbf{Edge Deletion:} 
    For every graph $G = (V, E, \phi, b)$ and edge $e^* \e (u,v) \in E$ it holds that
    \[
        F_{e} (V, E \setminus \{e^*\}, \phi|_{E \setminus \{e^*\}}, b) =F_{e} (G)
    \]
    for every $e \e (w,w') \in E$ such that $w \not\in S_u(G) \cup \{u\}$.
    %
    \item \textbf{Edge Multiplication:}     
    For every graph $G = (V, E, \phi, b)$ and edge $e^* \e (u,v) \in E$ such that $E_u^+(G) \subseteq E_v^-(G)$ let $E' = E \setminus E_u^+(G) \cup \{ e^*\}$. Then, it holds that
    \[
        F_e(V, E', \phi|_{E'}, b) =
        \begin{cases}
            \deg^+_u(G) \cdot F_{e^*}(G) & \text{if } e = e^*, \\ 
            F_e(G) & \text{otherwise,}
        \end{cases}
    \]
    for every $e\in E'$.
    \item \textbf{Edge Swap:} 
    For every graph $G = (V, E, \phi, b)$ and edges $e_1 \ee (u_1,v_1), e_2 \ee (u_2,v_2) \in E$ such that $F_{e_1} (G) = F_{e_2}(G)$ it holds that
    \[
        F_e(V, E, \phi[e_1 \rightarrow (u_1,v_2), e_2 \rightarrow (u_2,v_1)], b)  = 
        F_e(G)
    \] 
    for every $e \in E$.
    \item\textbf{Node Redirect:} 
    For every graph $G = (V, E, \phi, b)$ and out-twins $u,w \in V$ with the bijection $\psi$
    it holds that
    \[
        F_e(R_{u\rightarrow w } (G)) =
        \begin{cases}
            F_{e} (G)+ F_{\psi (e)} (G)&\text{if } e \in E^+_w(G), \\
            F_e (G)&\text{otherwise,}
        \end{cases}
    \]
    for every $e\in E \setminus E^+_u(G)$.
    %
    %
    \item \textbf{Baseline:}
    For every graph $G = (V,E,\phi, b)$ and isolated edge $e \e (u,v) \in E$ it holds that $F_e(G) = b(u)$.
\end{itemize}

These six axiom uniquely characterize Edge PageRank.

\begin{theorem}
\label{theorem:main}
Edge centrality measure $F$ satisfies Node Deletion, Edge Deletion, Edge Multiplication, Edge Swap, Node Redirect, and Baseline, if and only if,  it is Edge PageRank.
\end{theorem}

Our first two axioms describe when the centrality of an edge is not affected by a removal of an element of a graph, a node or an edge.
More in detail, \emph{Node Deletion} states that removing an isolated node does not affect the centrality of any edge.
In turn, \emph{Edge Deletion} says that removing an edge does not affect the centralities of edges that cannot be reached from the start of the removed edge.
Both axioms capture the intuition that the importance of an edge should not be affected by the parts of a network that are disconnected from it. 
They are direct adaptations of the axioms proposed in \cite{Was:Skibski:2020:pagerank}, only there, the axioms state that the centralities of respective nodes are not affected.

The next axiom, \emph{Edge Multiplication}, has been more significantly modified.
The original Edge Multiplication is focused on a node---it states that creating additional copies of all outgoing edges of a node does not affect the centralities of any node in a graph.
Here, we modify the axiom so that it focuses on an edge instead.
We say that if an edge is the only outgoing edge of a node, then creating $k-1$ additional copies of this edge, divides its centrality by $k$ and does not affect the centrality of other edges.
For the sake of notational convenience, in the formulation of the axiom we consider removing these additional copies of an edge, but the meaning is equivalent.
Intuitively, the axiom means that for a centrality measure the absolute number of edges is not important.

Our next axiom is \emph{Edge Swap}.
Here, again, the axiom is quite different from the one introduced for (node) PageRank.
In its original wording, Edge Swap considers swapping the ends of outgoing edges of two nodes with equal out-degrees and equal centralities.
Thus, this condition could not be included in our axiom since the centrality of a node is not defined.
To solve this issue, we require only that the swapped edges have equal centrality.
In this way, we obtain a simpler axiom: if two edges have equal centralities, swapping their ends does not affect the centrality of any edge.
It captures the key property of feedback centralities:
that your importance depends only on the importance of your incoming edges.

Intuitively, our next axiom, \emph{Node Redirect}, says that two nodes that are identical (with regards to outgoing edges) can be redirected to each other without affecting the centralities in a network.
It entails the same intuition as Node Redirect in \cite{Was:Skibski:2020:pagerank}, but the construction is a little more subtle.
The original axiom considered redirecting a node into its out-twin and stated that this operation sums up the centralities of the out-twins and does not affect the centralities of other nodes.
Here, we say that since there is a one-to-one correspondence between the outgoing edges of out-twins, their redirection should sum up the centralities of corresponding edges.
Also, the centralities of other edges in a graph should not be affected.

Finally, the original \emph{Baseline} axiom stated that centrality of an isolated node is equal to its weight.
The intuition behind it was that an isolated node is not influenced by the topology of a graph, so its centrality should be equal to its base importance that is reflected in its weight.
Hence, in our version of Baseline we consider an isolated edge 
and say that its centrality is equal to the weight of its start.

We conclude this section with an example that shows how our axiomatization can be used to decide whether Edge PageRank can be used in a particular application.

\begin{example}
Consider a network of flight connections between airports, in which an edge from airport A to airport B represents a single flight from A to B.
We can think about the importance of each connection on the global virus spread, as considered by~\citet{Marcelino:Kaiser:2012}.
Imagine that there is an airport, A, from which there are flights to only one other airport, B.
In such a case, Edge Multiplication implies that the change in the number of flights from A to B does not affect the importance of any edge.
However, additional flights from A to B increase the chance of virus spread from A to B, which in turn makes it more probable that flights outgoing from B would spread the virus.
Therefore, Edge Multiplication is not an adequate axiom in such a setting.
This implies that since Edge PageRank satisfies Edge Multiplication, it should not be used in this application.
\end{example}

\section{Proof of Uniqueness}
In this section, we sketch the proof of our main result that Edge PageRank is uniquely characterized by our axioms (Theorem~\ref{theorem:main}).
The full proof can be found in the appendix.

We begin by noting that Edge PageRank indeed satisfies all six axioms.
\begin{lemma}
\label{lemma:axioms}
Edge PageRank satisfies Node Deletion, Edge Deletion, Edge Multiplication, Edge Swap, Node Redirect, and Baseline.
\end{lemma}
We skip the proof of Lemma~\ref{lemma:axioms} and instead focus on proving that an arbitrary centrality measure $F$ satisfying our axioms is indeed Edge PageRank.
To this end, we adopt the following structure of the proof:

\begin{itemize}

\item First, we prove simple properties of $F$ (Proposition~\ref{prop:loc-sw-se}). Specifically, we show that the centrality of an edge depends only on its connected component (\emph{Locality}), does not depend on weights of sinks (\emph{Sink Weight}), and if the edge is a source edge, then it is equal to the weight of its start (\emph{Source Edge}).

\item Then, we analyze simple graphs in which all edges are incident to one node, $v$. 
\begin{itemize}
\item We begin with graphs in which $v$ has one incoming edge (from $u$) and one outgoing edge (to $w$).
First, we show that there exists a constant $a_F$ such that the centrality of the outgoing edge of $v$ equals $a_F$ times $b(u)$ plus $b(v)$ (Lemma~\ref{lemma:two_path_weighted}). Then, we prove that $a_F \in [0,1)$ (Lemma~\ref{lemma:a_f}).

\item Furthermore, we consider graphs in which node $v$ has $k$ outgoing and zero (Lemma~\ref{lemma:star}) or one (Lemma~\ref{lemma:star_edge}) incoming edges.
\end{itemize}    
\item Finally, we prove that for arbitrary graph, $F$ is equal to Edge PageRank with the decay factor $a_F$ (Lemma~\ref{lemma:any_graph}).

\end{itemize}

We note that there are significant differences between our proof and the proof of unique characterization of (node) PageRank introduced in~\cite{Was:Skibski:2020:pagerank}.
More in detail, the main axis of the proof for (node) PageRank is the induction on the number of cycles in a graph.
Here, we follow a different path and prove that centrality $F$ satisfying all our axioms satisfies also PageRank recursive equation (Equation~\eqref{eq:pagerank}).
Since this equation uniquely defines PageRank, this implies that $F$ is indeed PageRank.
In effect, the proof obtained in this paper is notable simpler.

We begin with a proposition that captures simple properties of an edge centrality measure implied by our axioms.
\begin{proposition}
\label{prop:loc-sw-se}
If edge centrality measure $F$ satisfies Node Deletion, Edge Deletion, Node Redirect, and Baseline, then for every graph $G = (V, E,\phi, b)$  and $e \in E$ it holds that
\begin{itemize}
    \item[(a)] (Locality) $F_e(G)=F_e(G + G')$ for every $G'$ disjoint with $G$,
    \item[(b)] (Sink Weight) $F_e(G)=F_e(V,E,\phi,b[w \rightarrow 0])$ for every sink $w \in V$,
    \item[(c)] (Source Edge) $F_e(G)=b(u)$ if $e \e (u,v)\in E$ is a source edge.
\end{itemize}
\end{proposition}
\begin{proof}[Proof (Sketch)]
    For (a), observe that
    from Edge Deletion, removing all edges from $G'$ in graph $G+G'$ does not affect the centrality of $e$.
    In resulting graph, all nodes from $G'$ are isolated.
    Thus, the thesis follows from Node Deletion.
    
    For (b), consider graphs
    $G''=(V,E,\phi,b[w \rightarrow 0])$
    and
    $G'=G'' + (\{w'\},\emptyset,\phi|_\emptyset,[b(w)])$.
    Then, $F_e(G'')=F_e(G')$ from \emph{(a)} and $F_e(G') = F_e(G)$ from Node Redirect since graph $G$ is graph $G'$ with $w'$ redirected into $w$.
    
    For (c), since $e$ is a source edge, then by Edge Deletion, removing all other edges does not affect its centrality.
    In result, $e$ is isolated, hence Baseline yields the thesis.
\end{proof}

In the following lemma, we introduce constant $a_F$ that, as we will prove later on, is decay factor of the Edge PageRank to which $F$ is equal to.

\begin{lemma}\label{lemma:two_path_weighted}
If edge centrality measure $F$ satisfies Node Deletion, Edge Deletion, Node Redirect, and Baseline, then there exists a constant, $a_F \in \mathbb{R}$, such that for every graph
$G \!=\! (\{u, v, w\},  \{e_1 \ee (u,v), e_2 \ee (v, w)  \},  [x, y, 0])$ it holds that
\[
    F_{e_1} (G) = x \quad \mbox{and} \quad F_{e_2} (G) = a_F \cdot x + y.
\]
\end{lemma}

\begin{proof}[Proof (Sketch)]
First equation follows directly from Source Edge (Proposition~\ref{prop:loc-sw-se}c).
Thus, we will focus on proving that $F_{e_2} (G) = a_F \cdot x + y$.
First, assume $y=0$ and consider graph
$G' \!\!=\! ( \{u'\!, v'\!, w'\}, \! \{e'_1 \ee (u' \! , v'), e'_2 \ee (v'\!, w')  \}, \! [x'\!, 0, 0])$
such that $u',v',w' \not \in \{u,v,w\}$
(see Figure~\ref{fig:two_path_weighted}).
By Locality (Proposition~\ref{prop:loc-sw-se}a), we know that in $G+G'$ the centrality of all edges is the same as in $G$ or $G'$.
In graph $G+G'$, we sequentially redirect nodes $w',v',u'$ into $w,v,u$ respectively, to obtain graph
$G'' \! \!=\! ( \{u,\! v,\! w\}, \! \{e_1 \ee (u,v), e_2 \ee (v, w)  \}, \! [x + x'\! , 0, 0])$,
which is exactly $G$, but with different weight of node $u$.
Since each time we redirected a node into its out-twin, from Node Redirect, we get that 
\[
    F_{e_2}(G'')=F_{e_2}(G)+F_{e'_2}(G').
\]
From this and the arbitrariness of the choice of nodes and edges, we conclude that there is a function, $f : \mathbb{R}_{\ge 0} \rightarrow \mathbb{R}_{\ge 0}$, such that $F_{e_2}(G)=f(x)$ and that $f$ is additive.
Since it is also non-negative, we get $f(x)=a_F \cdot x$~\cite{Cauchy:1821}.

\begin{figure}[t]
\centering
\begin{tikzpicture}[x=4cm,y=4cm] 
  \tikzset{     
    e4c node/.style={circle,draw,minimum size=0.6cm,inner sep=0},
    e4c edge/.style={sloped,above,font=\footnotesize},
    e4cn edge/.style={above,font=\footnotesize, pos=0.6}
  }
  
  \def\x{0};
  \def\y{0};
  \node (G) at (\x-0.15, \y+0.4) {$G$}; 
  \node (G_) at (\x-0.15, \y+0.1) {$G'$}; 
  \node[e4c node] (1) at (\x+0, \y+0.3) {$u$}; 
  \node[e4c node] (2) at (\x+0.3, \y+0.3) {$v$}; 
  \node[e4c node] (3) at (\x+0.6, \y+0.3) {$w$}; 
  \node[e4c node] (4) at (\x+0, \y+0) {$u'$}; 
  \node[e4c node] (5) at (\x+0.3, \y+0) {$v'$}; 
  \node[e4c node] (6) at (\x+0.6, \y+0) {$w'$}; 

  \path[->,draw,thick]
  (1) edge[e4c edge] node  {$e_1$}  (2) 
  (2) edge[e4c edge] node  {$e_2$}  (3)
  (4) edge[e4c edge] node  {$e'_1$} (5)
  (5) edge[e4c edge] node  {$e'_2$} (6)
  ;
  
  \def\x{1.1};
  \node (G) at (\x-0.15, \y+0.4) {$G^*$}; 
  \node[e4c node] (1) at (\x+0, \y+0.3) {$u$}; 
  \node[e4c node] (2) at (\x+0.3, \y+0.3) {$v$}; 
  \node[e4c node] (3) at (\x+0.6, \y+0.3) {$w$}; 
  \node[e4c node] (5) at (\x+0.3, \y+0) {$v'$};

  \path[->,draw,thick]
  (1) edge[e4c edge] node {$e_1$} (2) 
  (2) edge[e4c edge] node  {$e_2$}  (3)
  (5) edge[e4cn edge, pos=0.6, below] node [yshift=-2] {$e_3$}  (3)
  ;

\end{tikzpicture}
\caption{An illustration to the proof of Lemma~\ref{lemma:two_path_weighted}.}
\label{fig:two_path_weighted}
\end{figure}
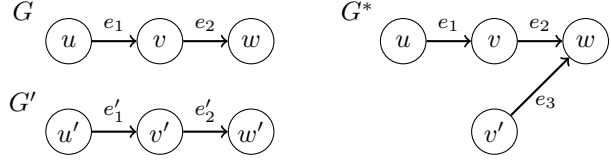

Now, if $y$ is not necessarily 0, then we consider graph
$G^* \! \!=\! ( \{u,\! v,\! v' \!,\! w\}, \! \{e_1 \ee (u,\! v), e_2 \ee (v,\! w), e_3 \ee (v' \!,\! w)  \}, \! [x, 0, y, 0])$ (see Figure~\ref{fig:two_path_weighted}).
Node $v$ is not a successor of $v'$, hence from Edge Deletion, Node Deletion, and the first part of the proof, we get that $F_{e_2}(G^*)=a_F \cdot x$.
Also, $F_{e_3}(G^*)=y$ from Source Edge (Proposition~\ref{prop:loc-sw-se}c).
Thus, since redirecting $v'$ into $v$ in $G^*$ results in $G$, Node Redirect yields the thesis.
\end{proof}

Now, let us show the bounds for constant $a_F$.


  



\begin{lemma}\label{lemma:a_f}
If edge centrality measure $F$ satisfies Node Deletion, Edge Deletion, Edge Swap, Node Redirect, and Baseline, then $a_F \in [0,1)$.
\end{lemma}
\begin{proof}[Proof (Sketch)]
Inequality $a_F \ge 0$ follows directly from Lemma~\ref{lemma:two_path_weighted} and the fact that centrality is non-negative.
Thus, let us focus on showing that $a_F<1$.
To this end, consider graph
$G=(\{u,v,w\},\{e_1 \ee (u,w), e_2 \ee (v,v) \},[x,1,0])$ (see Figure~\ref{fig:a_f}).
From Locality (Proposition~\ref{prop:loc-sw-se}a), $F_{e_2}(G)$ does not depend on $x$.
If we take $x = F_{e_2}(G)$, then from Source Edge (Proposition~\ref{prop:loc-sw-se}c) $F_{e_1}(G)= x = F_{e_2}(G)$.
In such case, from Edge Swap, exchanging the ends of these two edges does not affect their centralities.
Hence, in resulting graph,
$G'=(\{u,v,w\},\{e_1 \ee (u,v), e_2 \ee (v,w) \},[x,1,0])$, we have
$F_{e_2}(G')=F_{e_2}(G)$.
On the other hand, from Lemma~\ref{lemma:two_path_weighted}, we have
$F_{e_2}(G')=1 + a_F \cdot x$.
We took $x=F_{e_2}(G)$, thus
\(
    F_{e_2}(G) = F_{e_2}(G')= 1 + a_F \cdot F_{e_2}(G),
\)
which implies that
\(
    F_{e_2}(G) = 1/(1-a_F).
\)
Since centrality is non-negative, we get $a_F<1$.
\end{proof}

In Lemmas~\ref{lemma:star} and \ref{lemma:star_edge} we consider graphs in which $v$ has multiple outgoing edges.
First, we assume that it does not have any incoming edge.

\begin{lemma}\label{lemma:star}
If edge centrality measure $F$ satisfies Node Deletion, Edge Deletion, Edge Multiplication, Node Redirect, and Baseline, then for every $k\in \mathbb{N}$ and graph 
$G \!=\! (\!\{v, w_1,\dots, w_k\},\! \{e_1 \ee (v,w_1), \dots, e_k \ee (v,w_k) \!\}, \! [x, 0, \dots, 0])$,
it holds that
\[
    F_{e_i}(G) = x / k \quad \mbox{for every } i\in \{1,\dots, k\}.
\]
\end{lemma}

\begin{proof}
Let us fix arbitrary $i \in \{ 1, \dots, k\}$.
Since nodes $w_1,\dots,w_k$ are all sinks, they are also out-twins.
Hence, from Node Redirect, we know that sequentially redirecting nodes $w_2,\dots,w_k$ into node $w_1$ preserves the centrality of edges $e_1,\dots,e_k$.
Therefore, in the obtained graph,
\(
    G' \! = ( \{v, w_1\},\! \{e_1 \ee (v,w_1), \dots, e_k \ee (v,w_1) \}, \! [x, 0]),
\)
we have $F_{e_i}(G) = F_{e_i}(G')$.
See Figure~\ref{fig:star} for an illustration.

Next, consider graph
\(
    G_i = ( \{v, w_1\},\!\{e_i \e (v,w_1) \},\! [x, 0]).
\)
Observe that $G_i$ can be obtained from $G'$ by removing all edges but $e_i$.
Thus, from Edge Multiplication, we get $F_{e_i}(G')=F_{e_i}(G_i) / k$.
On the other hand, from Baseline, $F_{e_i}(G_i)=x$.
Combining all equations, we obtain that
\(
    F_{e_i}(G) = F_{e_i}(G') = F_{e_i}(G_i) / k = x/k.
\)
\end{proof}



\begin{figure}[t]
\centering
\begin{minipage}{.5\textwidth}
    \centering
    \begin{tikzpicture}[x=4cm,y=4cm] 
      \tikzset{     
        e4c node/.style={circle,draw,minimum size=0.6cm,inner sep=0}, 
        e4c edge/.style={sloped,above,font=\footnotesize},
        e4cn edge/.style={above,pos=0.6,font=\footnotesize}
      }
      \def\x{0}
      \def\h{0.3}
      \node (G) at (\x-0.1, \h) {$G$}; 
      \node[e4c node] (1) at (\x+0.2, \h) {$u$}; 
      \node[e4c node] (2) at (\x+0.00, 0) {$w$}; 
      \node[e4c node] (3) at (\x+0.40, 0) {$v$};

      \path[->,draw,thick]
      (3) edge[e4c edge, loop above] node  {$e_2$} (3)
      (1) edge[e4cn edge]  node  [xshift=-3]  {$e_1$} (2)
      ;
      
      \def\x{0.8}
      \node (G) at (\x-0.1, \h) {$G'$}; 
      \node[e4c node] (1) at (\x+0.2, \h) {$u$}; 
      \node[e4c node] (2) at (\x+0.00, 0) {$w$}; 
      \node[e4c node] (3) at (\x+0.40, 0) {$v$};

      \path[->,draw,thick]
      (3) edge[e4c edge] node {$e_2$} (2)
      (1) edge[e4cn edge]  node  [xshift=3] {$e_1$} (3)
      ;
    
    \end{tikzpicture}
    \captionof{figure}{An illustration to the proof of Lemma~\ref{lemma:a_f}.}
    \label{fig:a_f}
\end{minipage}%
\begin{minipage}{.5\textwidth}
    \centering
    \begin{tikzpicture}[x=4cm,y=4cm] 
      \tikzset{     
        e4c node/.style={circle,draw,minimum size=0.6cm,inner sep=0},
        e4c edge/.style={right,font=\footnotesize},
        e4cn edge/.style={above,font=\footnotesize}
      }
      
      \def\xdist{0.25}
      \def\ydist{0.3}
      \def\b{20} 
      
      \def\x{0};
      \def\y{0};
      \node (G) at (\x, \y+\ydist + 0.1) {$G$};
      \node[e4c node] (2) at (\x+\xdist, \y+\ydist) {$v$};
      \node[e4c node] (3) at (\x+0, \y) {$w_1$}; 
      \node[e4c node] (4) at (\x+\xdist, \y) {$w_2$};
      \node[e4c node] (5) at (\x+2*\xdist, \y) {$w_3$}; 
    
      \path[->,draw,thick]
      (2) edge[e4c edge] node [left]  {$e_1$}  (3)
      (2) edge[e4c edge] node  [xshift=-2.5]  {$e_2$}  (4)
      (2) edge[e4c edge] node  {$e_3$} (5)
      ;
    
      \def\x{0.65};
      \node (G) at (\x + \xdist - 0.2, \y+\ydist + 0.1) {$G'$};
      \node[e4c node] (2) at (\x+\xdist, \y+\ydist) {$v$};
      \node[e4c node] (3) at (\x+\xdist, \y) {$w_1$};
      
      \path[->,draw,thick]
      (2) edge[e4c edge, bend right=50]  node [left]  {$e_1$} (3)
      (2) edge[e4c edge] node  [xshift=-2.5]  {$e_2$} (3)
      (2) edge[e4c edge, bend left=50] node  {$e_3$} (3)
      ;

      \def\x{1.1};
      \node (G) at (\x + \xdist - 0.2, \y+\ydist + 0.1) {$G_i$};
      \node[e4c node] (2) at (\x+\xdist, \y+\ydist) {$v$};
      \node[e4c node] (3) at (\x+\xdist, \y) {$w_1$};
      
      \path[->,draw,thick]
      (2) edge[e4c edge]  node  {$e_i$}  (3)
      ;
    \end{tikzpicture}
    \captionof{figure}{An illustration to the proof of Lemma~\ref{lemma:star} for $k=3$.}
    \label{fig:star}
\end{minipage}
\end{figure}

Now, we assume that $v$ has exactly one incoming edge and multiple outgoing edges.

\begin{lemma}\label{lemma:star_edge}
If edge centrality measure $F$ satisfies Node Deletion, Edge Deletion, Edge Multiplication, Node Redirect, and Baseline, then for every $k\in \mathbb{N}$ and graph 
\[\begin{split}
G = (&\{u, v, w_1,\dots w_k\}, \\
&\{e \e (u,v), e_1 \e (v,w_1), \dots, e_k \e (v,w_k)\}, \\
&[x, y, 0, \dots, 0]), 
\end{split}\]
it holds that $F_e(G) = x$ and $F_{e_i}(G) = (a_F \cdot x + y) / k$ for every $i\in \{1,\dots, k\}$.
\end{lemma}
\begin{proof}[Proof (Sketch)]
The proof is analogous to the proof of Lemma~\ref{lemma:star} with only two differences:
first, in every graph node $v$ has one incoming edge from $u$;
second, the final equation for graphs $G_i$ is of the form $F_{e_i}(G_i)= a_F \cdot x + y$, which is implied by Lemma~\ref{lemma:two_path_weighted}. 
\end{proof}

We are now ready for the final lemma of this section.

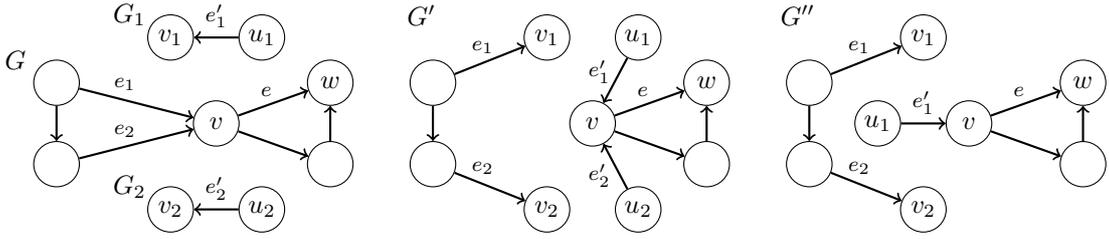
\begin{figure*}[t]
\centering
\begin{tikzpicture}[x=3cm,y=3cm] 
  \tikzset{     
    e4c node/.style={circle,draw,minimum size=0.6cm,inner sep=0}, 
    e4c edge/.style={above,font=\footnotesize}
  }
  \def\d{0}
  \def\shift{0.1}
  \def\x{0.4}
  \def\xx{0.6}
  \def\y{0.18}
  \def\yy{0.38}
  \node (G) at (\d - 0.18, \y + 0.1) {$G$};
  \node (G1) at (\d + \shift + \x - 0.18, \yy + 0.1) {$G_1$};
  \node (G2) at (\d + \shift + \x - 0.18, -\yy + 0.1) {$G_2$};
  \node[e4c node] (u_1) at (\d, \y) {}; 
  \node[e4c node] (u_2) at (\d, -\y) {};
  \node[e4c node] (v1) at (\d + \shift + \x, \yy) {$v_1$};
  \node[e4c node] (v2) at (\d + \shift + \x, -\yy) {$v_2$};
  \node[e4c node] (v) at (\d + \shift + \xx, 0) {$v$};
  \node[e4c node] (u1) at (\d + \shift + 2*\x, \yy) {$u_1$};
  \node[e4c node] (u2) at (\d + \shift + 2*\x, -\yy) {$u_2$};
  \node[e4c node] (w_1) at (\d + 2*\xx, \y) {$w$};
  \node[e4c node] (w_2) at (\d + 2*\xx, -\y) {};
  
  \path[->,draw,thick, e4c edge]
  (u_1) edge node[above, pos=0.4] {$e_1$} (v)
  (u_1) edge (u_2)
  (u_2) edge node[above, pos=0.4] {$e_2$} (v)
  (v) edge node[above, pos=0.4] {$e$}  (w_1)
  (v) edge (w_2)
  (w_2) edge (w_1)
  (u1) edge node[above] {$e'_1$} (v1)
  (u2) edge node[above] {$e'_2$} (v2)
  ;

  \def\d{1.65}
  \node (G) at (\d -0.05, \yy + 0.1) {$G'$};
  \node[e4c node] (u_1) at (\d, \y) {}; 
  \node[e4c node] (u_2) at (\d, -\y) {};
  \node[e4c node] (v1) at (\d + \shift + \x, \yy) {$v_1$};
  \node[e4c node] (v2) at (\d + \shift + \x, -\yy) {$v_2$};
  \node[e4c node] (v) at (\d + \shift + \xx, 0) {$v$};
  \node[e4c node] (u1) at (\d + \shift + 2*\x, \yy) {$u_1$};
  \node[e4c node] (u2) at (\d + \shift + 2*\x, -\yy) {$u_2$};
  \node[e4c node] (w_1) at (\d + 2*\xx, \y) {$w$};
  \node[e4c node] (w_2) at (\d + 2*\xx, -\y) {};
  
  \path[->,draw,thick, e4c edge]
  (u_1) edge node[above, pos=0.4] {$e_1$} (v1)
  (u_1) edge (u_2)
  (u_2) edge node[above, pos=0.4] {$e_2$} (v2)
  (v) edge node[above, pos=0.4] {$e$}  (w_1)
  (v) edge (w_2)
  (w_2) edge (w_1)
  (u1) edge node[left, pos=0.3] {$e'_1$} (v)
  (u2) edge node[left, pos=0.3] {$e'_2$} (v)
  ;

  \def\d{3.3}
  \node (G) at (\d -0.05, \yy + 0.1) {$G''$};
  \node[e4c node] (u_1) at (\d, \y) {}; 
  \node[e4c node] (u_2) at (\d, -\y) {};
  \node[e4c node] (v1) at (\d + \shift + \x, \yy) {$v_1$};
  \node[e4c node] (v2) at (\d + \shift + \x, -\yy) {$v_2$};
  \node[e4c node] (v) at (\d + \shift + \xx, 0) {$v$};
  \node[e4c node] (u1) at (\d + 0.5*\xx, 0) {$u_1$};
  \node[e4c node] (w_1) at (\d + 2*\xx, \y) {$w$};
  \node[e4c node] (w_2) at (\d + 2*\xx, -\y) {};
  
  \path[->,draw,thick, e4c edge]
  (u_1) edge node[above, pos=0.4] {$e_1$} (v1)
  (u_1) edge (u_2)
  (u_2) edge node[above, pos=0.4] {$e_2$} (v2)
  (v) edge node[above, pos=0.4] {$e$}  (w_1)
  (v) edge (w_2)
  (w_2) edge (w_1)
  (u1) edge node[above] {$e'_1$} (v)
  ;
\end{tikzpicture}
\caption{The scheme of the proof of Lemma~\ref{lemma:any_graph} for an example graph, $G$, with $m=2$.}
\label{fig:any_graph}
\end{figure*}

\begin{lemma}\label{lemma:any_graph}
If edge centrality measure $F$ satisfies Node Deletion, Edge Deletion, Edge Multiplication, Edge Swap, Node Redirect, and Baseline, then $F$ is Edge PageRank.
\end{lemma}
\begin{proof}[Proof (Sketch)]
We will prove that $F$ satisfies Edge PageRank recursive equation (Equation~\eqref{eq:pagerank}) with decay factor $a_F$ for every graph $G = (V, E,\phi, b)$  and edge $e \e (v,w) \in E$.
Since this equation uniquely define Edge PageRank, it will imply that $F$ is indeed Edge PageRank.

First assume that $v$ does not have incoming edges.
Then, removing all edges not incident with $v$ and nodes that become isolated in doing so, we obtain a graph from Lemma~\ref{lemma:star} except possibly non-zero weights of sinks.
Hence, from Edge Deletion, Node Deletion, Sink Weight (Proposition~\ref{prop:loc-sw-se}b), and Lemma~\ref{lemma:star}, we get $F_e(G)=b(v)/\deg^+_v(G)$, which is Equation~\eqref{eq:pagerank} for edge $e$ in graph $G$.

Thus, let us assume that $v$ has $m>0$ incoming edges, i.e., $E^-_v(G) = \{e_1,\dots,e_m$\}. 
We denote their centralities by $x_i = F_{e_i}(G)$ for every $i \in \{1,\dots,m\}$.
In what follows, through several operations we transform graph $G$ into graph from Lemma~\ref{lemma:star_edge} such that the centrality of $e$ is unchanged.

To this end, for every $i \in \{1,\dots,m\}$, we add to $G$ a simple one-edge graph $G_i=( \{u_i,v_i\},\{e'_i \ee (u_i,v_i)\},[x_i,0])$
(see Figure~\ref{fig:any_graph}).
In the obtained sum of graphs, for every $i \in \{1,\dots,m\}$, edges $e_i$ and $e'_i$ have equal centralities
(from Locality (Proposition~\ref{prop:loc-sw-se}a) and Baseline).
Thus, if we exchange their ends, then by Edge Swap we won't affect the centrality of any edge.
Exchanging ends sequentially for all such pairs, we obtain graph $G'$ in which incoming edges of $v$, i.e., $E^-_v(G')=\{e'_1,\dots,e'_m\}$, are all source edges.

Now, observe that in graph $G'$, nodes $u_1,\dots,u_m$ are out-twins (each has one outgoing edge to $v$).
Thus, from Node Redirect we can sequentially redirect nodes $u_2,\dots,u_m$ into $u_1$ without affecting the centrality of $e$.
Let us denote resulting graph by $G''$.
Observe that the weight of node $u_1$ in graph $G''$ is equal to $\sum_{i=1}^m x_i$.
Also, node $v$ has one incoming edge, $e'_1$, that is a source edge.
Removing all edges not incident with $v$ and nodes that become isolated in doing so, we obtain graph from Lemma~\ref{lemma:star_edge} but with possibly non-zero weights of sinks.
Thus, from Edge Deletion, Node Deletion, Sink Weight (Proposition~\ref{prop:loc-sw-se}b), and Lemma~\ref{lemma:star_edge} we get that
\[
    F_e(G'') = \frac{1}{\deg^+_v(G'')} \left( a_F \cdot \sum_{i=1}^m x_i + b(v) \right).
\]
Since $\sum_{i=1}^m x_i =  \sum_{e_i \in E^-_v(G)}  F_{e_i}  (G)$, $F_e(G'')=F_e(G)$, and $\deg^+_v  (G'')  =  \deg^+_v (G)$,
Equation~\eqref{eq:pagerank} holds.
\end{proof}


  


\section{Comparison with Other Edge Centralities}
In this section, we provide an overview of edge centrality measures from the literature and analyze which of our axioms they satisfy.
Some of these measures are defined only for a specific class of graphs, e.g., strongly connected graphs.
In such case, we consider our axioms \emph{restricted} to this class, i.e., we add an additional constraint that in all graphs considered in the axiom the centrality is well defined.

\emph{Eigenedge}~\cite{Huang:Huang:2019}
assumes that the centrality of an edge is proportional to the sum of the centralities of the edges incoming to its start.
Formally, it is defined as a measure that for every strongly connected graph $G=(V,E,\phi,b)$ and edge $e \e (u,v) \in E$ satisfies
\[
    Eig_e(G) = \frac{1}{\lambda} \cdot \sum_{e' \in E^-_u(G)} Eig_{e'}(G),
\]
where $\lambda$ is the largest eigenvalue of the adjacency matrix of $G$.
Usually, a normalization condition is added to make the solution unique, e.g., that the sum of all centralities is equal to 1.
With this condition, Eigenedge is well defined for all strongly connected graphs and 
satisfies all our axioms restricted to this class except for Edge Multiplication.

\emph{Edge Katz centrality} can be derived from (node) Katz centrality~\cite{Katz:1953} in the same way as Eigenedge is derived from Eigenvector centrality~\cite{Bonacich:1972}, and Edge PageRank from (node) PageRank.
It works similarly to Eigenedge, but to every edge we add a base centrality equal to the weight of its start.
Formally, for a decay factor $a \in \mathbb{R}_{\ge 0}$ it is defined as a unique measure that for every graph $G=(V,E,\phi,b)$ with $\lambda < 1/a$ and edge $e \e (u,v) \in E$ satisfies
\[
    K^a_e(G) = a \cdot \sum_{e' \in E^-_u(G)} K^a_{e'}(G) + b(u).
\]
Katz centrality satisfies all our axioms restricted to the class of graphs with $\lambda < 1/a$ except for Edge Multiplication.

In the same way, \emph{Edge Seeley index} can be derived from (node) Seeley index~\cite{Seeley:1949} (known also as \emph{Katz Prestige} or \emph{simplified PageRank}).
It can be seen as a borderline case of PageRank, when we increase decay factor $a$ to 1~\cite{Was:Skibski:2021:feedback}.
Formally, it is defined as the measure that for every strongly connected graph $G=(V,E,\phi,b)$ and edge $e \e (u,v) \in E$ satisfies
\[
    SI_e(G) = \frac{1}{\deg^+_{u}(G)} \sum_{e' \in E^-_u(G)} SI_{e'}(G).
\]
Like for Eigenedge, to obtain uniqueness, we add a normalization condition that the sum of centralities is equal to 1.
Then, Seeley index satisfies all of our axioms restricted to the class of strongly connected graphs
(since axioms are restricted, this does not contradict Theorem~\ref{theorem:main}).

\emph{Edge Betweenness}~\cite{Girvan:Newman:2002} measures how often an edge is on a shortest path between two nodes.
For graph $G = (V,E,\phi,b)$ and edge $e \in E$ it is defined as
\[
    B_e(G) = \sum_{s,t\in V, s \neq t, t\in S_s(G)}\frac{\delta_{s,t}(e)}{\delta_{s,t}},
\]
where $\delta_{s,t}$ is the number of shortest paths from $s$ to $t$ and $\delta_{s,t}(e)$ is the number of such paths passing through $e$.
Edge Betweenness satisfies only Node Deletion.

\emph{Information centrality}~\cite{Fortunato:etal:2004} is defined as the relative loss in the network efficiency
that results from removing an edge.
Formally, for every graph $G=(V,E,\phi,b)$ and edge $e \in E$ we have
\[
    I_e(G) \!=\! \left(E\! f\! f(G) \!-\! E\! f\! f(V, E\setminus \{e\}, \phi|_{E\setminus \{e\}}, b)\right) \!/ E\! f\! f(G),
\]
where $E\! f\! f(G) = \sum_{v,w\in V, v\not= w} 1/dist_{v,w}(G)$.
It is well defined for strongly connected graphs and satisfies Node Deletion, Edge Deletion, and Baseline restricted to this class.

\emph{GTOM (Generalized Topological Overlap Matrix)}~\cite{Yip:Horvath:2007} of an edge measures how many common direct successors have its start and its end.
Formally, for $G=(V,E,\phi,b)$ and edge $e \e (u,v) \in E$ it is defined as
\[
    GT O\! M_e(G) =
        \frac{|\{w \! : \! (u, \! w),(v,\! w) \! \in \! E \}| + 1}
        {\min \! \big(| \{w \! : \! (u,\! w) \! \in \! E\} |,  | \{w \! : \! (v,\! w) \! \in \! E\} | \big)}.
\]
It satisfies only Node Deletion.

We do not consider \emph{Spanning Edge Betweenness}
~\cite{Teixeira:etal:2013} as it is defined only on undirected graphs.

\section{Using Line Graphs in Axiomatization}

Many edge centrality measures, including Edge PageRank, can be equivalently defined as node centralities in line graphs.
In this section, we analyze whether this fact can be used in creating axiomatization of an edge centrality based on the axiomatization of a node centrality.

For a graph $G$, the \emph{line graph} $L(G)$ is a graph that represents adjacencies between edges of $G$. 
Specifically, nodes of the line graph are edges of $G$ and edges of the line graph connect nodes that represent edges incident in $G$.
Formally, for graph $G = (V,E,\phi,b)$, its line graph is defined as
\[ L(G) = (E, \{(e_i, e_j) : \phi_2(e_i) = \phi_1(e_j)\}, b'), \]
with $b'(e) = b(u)/\deg^+_{u}(G)$ for every $e : (u,v) \in E$.
See Figure~\ref{fig:line_graph} for an illustration.

\citet{Chapela:etal:2015} proved that Edge PageRank is equivalent to PageRank of the corresponding line graph: $PR^a_e(L(G)) = PR^a_e(G)$ for every graph $G$ and edge $e$.
This result suggests that an axiomatization can be obtained by using node centrality axioms on line graphs.
Observe that most axioms proposed for (node) PageRank are so-called \emph{invariance axioms}: they specify graph operations that do not change centralities of nodes.
For example, \emph{Edge Swap} states that the ends of edges from equally important nodes with equal out-degrees can be swapped.
Can we just use these axioms to uniquely characterize (node) PageRank on line graphs which corresponds to Edge PageRank?

As it turns out, it is not possible. 
The first reason is the fact that line graphs are not closed under some of the invariance operations.
To see that, consider graph $G$ from Figure~\ref{fig:line_graph}.
Since graph $L(G)$ is symmetrical, by Edge Swap, we can swap the ends of edges $(e_2,e_6)$ and $(e_3,e_7)$.
However, we can prove that the graph obtained in this way is no longer a line graph of any graph.
%
Assume that it is the line graph of some $G'=(V',E',\phi',b')$.
Since after edge swap, $e_1$ and $e_3$ have edges to $e_6$, it means that $\phi'_2(e_1)=\phi'_1(e_6)=\phi'_2(e_3)$.
Similarly, we get that $\phi'_2(e_3)=\phi'_1(e_8)=\phi'_2(e_4)$.
Thus, we get $\phi'_2(e_1)=\phi'_2(e_4)$.
However, this implies that in the line graph there are edges $(e_4,e_6)$ and $(e_1,e_8)$, none of which is present in our graph---a contradiction.


%


Finally, we note that not every edge centrality can be defined as a node centrality of a line graph.
To see that, observe that merging starts of edges $e_1$ and $e_4$ in graph $G$ from Figure~\ref{fig:line_graph} does not affect the line graph.
However, for some edge centralities, e.g., Edge Betweenness, centrality of $e_4$ changes after such merge.
That is why analyzing edge centralities cannot be reduced to the analysis of node centralities in line graphs.

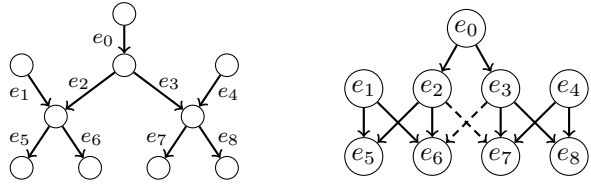
\begin{figure}[t]
\centering
\begin{tikzpicture}[x=3cm,y=2cm] 
  \tikzset{     
    e4c node/.style={circle,draw,minimum size=0.3cm,inner sep=0},
    e4c edge/.style={left,font=\footnotesize}
  }
  
  \def\xdist{0.3}
  \def\ydist{0.34}
  \def\nodelettersize{\scriptsize}
  
  \node[e4c node] (1) at (1.5*\xdist, 3*\ydist) {}; 
  \node[e4c node] (10) at (0*\xdist, 2*\ydist) {}; 
  \node[e4c node] (2) at (1.5*\xdist, 2*\ydist) {}; 
  \node[e4c node] (8) at (3*\xdist, 2*\ydist) {}; 
  \node[e4c node] (3) at (0.5*\xdist, \ydist) {}; 
  \node[e4c node] (4) at (2.5*\xdist, \ydist) {}; 
  \node[e4c node] (9) at (0.0*\xdist, 0) {}; 
  \node[e4c node] (7) at (1.0*\xdist, 0) {}; 
  \node[e4c node] (5) at (2.0*\xdist, 0) {}; 
  \node[e4c node] (6) at (3.0*\xdist, 0) {};

  \path[->,draw,thick]
  (1) edge[e4c edge] node {$e_0$} (2)
  (10) edge[e4c edge,pos=0.6,left] node {$e_1$} (3)
  (2) edge[e4c edge,pos=0.3] node {$e_2$} (3)
  (2) edge[e4c edge,pos=0.3,right] node {$e_3$} (4)
  (8) edge[e4c edge,pos=0.6,right] node {$e_4$} (4)
  (3) edge[e4c edge,pos=0.4] node {$e_5$} (9)
  (3) edge[e4c edge,pos=0.4,right] node {$e_6$} (7)
  (4) edge[e4c edge,pos=0.4] node {$e_7$} (5)
  (4) edge[e4c edge,pos=0.4,right] node {$e_8$} (6)
  ;

  \tikzset{     
    e4c node/.style={circle,draw,minimum size=0.5cm,inner sep=0},
    e4c edge/.style={left,font=\footnotesize}
  }
  \def\x{1.5}
  \def\y{-0.15}
  \def\xdist{0.3}
  \def\ydist{0.45}

  \node[e4c node] (1) at (\x+1.5*\xdist, \y+2.5*\ydist-0.05) {$e_0$}; 
  \node[e4c node] (10) at (\x+0*\xdist, \y+1.5*\ydist) {$e_1$}; 
  \node[e4c node] (2) at (\x+1*\xdist, \y+1.5*\ydist) {$e_2$}; 
  \node[e4c node] (3) at (\x+2*\xdist, \y+1.5*\ydist) {$e_3$}; 
  \node[e4c node] (4) at (\x+3*\xdist, \y+1.5*\ydist) {$e_4$}; 
  \node[e4c node] (8) at (\x+0*\xdist, \y+0.5*\ydist) {$e_5$}; 
  \node[e4c node] (7) at (\x+1*\xdist, \y+0.5*\ydist) {$e_6$}; 
  \node[e4c node] (5) at (\x+2*\xdist, \y+0.5*\ydist) {$e_7$}; 
  \node[e4c node] (6) at (\x+3*\xdist, \y+0.5*\ydist) {$e_8$};

  \path[->,draw,thick]
  (10) edge[e4c edge]  (7)
  (10) edge[e4c edge]  (8)
  (2) edge[e4c edge]  (8)
  (1) edge[e4c edge]  (3)
  (3) edge[e4c edge]  (5)
  (4) edge[e4c edge]  (5)
  (1) edge[e4c edge]  (2)
  (2) edge[e4c edge]  (7)
  (4) edge[e4c edge]  (6)
  (3) edge[e4c edge]  (6)
  (2) edge[e4c edge, densely dashed]  (5)
  (3) edge[e4c edge, densely dashed]  (7)
  ;
\end{tikzpicture}

\caption{Graph $G$ (on the left-hand side) and its line-graph $L(G)$ (on the right-hand side). Dashed edges represent the effect of the edge swap operation.}\label{fig:line_graph}
\end{figure}

\section{Related Work}

Axiomatic characterizations have been proposed for many node centrality measures~\cite{Garg:2009,Skibski:etal:2019:attachment}, including both simplified PageRank~\cite{Altman:Tennenholtz:2005, Palacios-Huerta:Volij:2004}, and PageRank in its general form~\cite{Was:Skibski:2020:pagerank}.
However, to the best of our knowledge, there are no axiomatic characterizations of edge centrality measures in the literature to date.

Several papers studied properties of Edge PageRank.
\citet{Chapela:etal:2015} proved that PageRank of an edge is equal to PageRank of the corresponding node in the line graph.
\citet{Kim:etal:2010} showed how Edge PageRank (under the name \emph{LinkRank}) can be used in community detection.
A similar, but different, edge metric derived from PageRank was used by \citet{Chung:Zhao:2010} in their graph sparsification algorithm.





\section{Conclusions}

In this paper, we proposed the first axiomatic characterization of an edge centrality measure in the literature.
Specifically, we proved that Edge PageRank is a unique centrality measure that satisfies six axioms: Node Deletion, Edge Deletion, Edge Multiplication, Edge Swap, Node Redirect, and Baseline.

Our paper initiates the research on axiomatic properties of edge centrality measures, many of which deserve their own characterizations.
An extension of our work into another direction would be to axiomatically analyze node similarity measures, e.g., SimRank.

\section{Acknowledgments}
This work was supported by the Polish National Science Center grant  2018/31/B/ST6/03201.
Natalia Kucharczuk was supported by the Ministry of Science and Higher Education project Szkola Orlow, project number 500-D110-06-0465160.

\printbibliography

\clearpage

\appendix
\renewcommand\thefigure{\arabic{figure}}
\setcounter{figure}{0}    


\section{Proof of Uniqueness}
In this section, we present the full proof of Theorem~\ref{theorem:main} that Edge PageRank is a unique edge centrality measure satisfying Node Deletion, Edge Deletion, Edge Multiplication, Edge Swap, Node Redirect, and Baseline.

\subsection{Edge PageRank Implies Axioms}
We begin by showing that Edge PageRank indeed satisfies all our axioms.

\begin{apptheorem}{Lemma}{\ref{lemma:axioms}} 
Edge PageRank satisfies Node Deletion, Edge Deletion, Edge Multiplication, Edge Swap, Node Redirect, and Baseline.
\end{apptheorem}
\begin{proof}
Take arbitrary constant $a \in [0,1)$ and graph $G=(V,E,\phi,b)$ and let us consider axioms one by one.
For most of our axioms, we consider graph $G'$ which is graph $G$ modified in the way described by the axiom and centrality in this graph, $F(G')$, that is $PR^a(G)$ modified in accordance with the axiom.
Then, we prove that $F(G')$ satisfies PageRank recursive equation (Equation~\eqref{eq:pagerank}) for graph $G'$.
Since this equation uniquely defines PageRank, this implies that $F(G')$ is in fact PageRank of $G'$ and thus PageRank indeed satisfies given axiom.

\subsubsection{Node Deletion}
Consider isolated node $u \in V$ and graph $G' = (V \setminus \{u\},E,\phi,b|_{V \setminus \{u\}})$.
Let centrality $F$ be such that $F_e(G')=PR^a_e(G)$ for every $e \in E$.
We will show that $F(G')$ satisfies PageRank recursive equation (Equation~\eqref{eq:pagerank}) for graph $G'$, 
which will imply that $F(G')=PR^a(G')$ and thus $PR^a_e(G')=PR^a_e(G)$ for every $e \in E$.

Fix arbitrary $e \in E$.
From Equation~\eqref{eq:pagerank} for PageRank in graph $G$ and the fact that $F_e(G')=PR^a_e(G)$ for every $e \in E$ we have
\begin{equation}
    \label{eq:axioms:nd}
    F_{e}(G') = \frac{1}{\deg_{u}^+(G)} \left( a\cdot\sum_{e' \in E^-_u(G)} F_{e'}(G') + b(u)\right) \!\!.
\end{equation}
Since no edges have been modified, $E^-_u(G) = E^-_u(G')$ and $\deg_{u}^+(G)=\deg_u^+(G')$.
Thus, Equation~\eqref{eq:axioms:nd} is equivalent to PageRank recursive equation for centrality $F(G')$ and edge $e$.
Therefore, $F(G')$ satisfies PageRank recursive equation and Node Deletion follows.

\subsubsection{Edge Deletion}
Let us consider edge $e^* \e (u,v) \in E$ and graph $G' = (V, E',\phi |_{E'},b)$, where $E' = E \setminus \{e^*\}$.
Observe that in contrary to Node Deletion, Edge Deletion does not specify the centrality of every edge in graph $G'$, only those with the start that is not a successor of $u$ or $u$ itself, i.e., $E'' = \{ e \in E' : \phi_1(e) \not \in S_u(G) \cup \{u\}\}$.
Thus, let us first construct centrality $F(G')$ for our purpose.
To this end, consider graph $G'$ induced by $u$ and the successors of $u$, i.e.,
$G''=(S_u(G) \cup \{u\}, E' \setminus E'',\phi|_{E' \setminus E''},b'')$ where 
$$b''(v)=b(v) + a \cdot \sum_{e \in E^-_v(G') \cap E''} PR^a_e(G).$$
Then, let us define centrality $F(G')$ as
\begin{equation}
    \label{eq:axioms:ed:1}
    F_e(G') =
    \begin{cases}
        PR^a_e(G) & \mbox{if } e \in E'',\\
        PR^a_e(G'') & \mbox{otherwise,}
    \end{cases}
\end{equation}
for every $e \in E'$.
We will show that $F(G')$ satisfies PageRank recursive equation (Equation~\eqref{eq:pagerank}) for graph $G'$, 
which will imply that $F(G')=PR^a(G')$ and thus $PR^a_e(G')=PR^a_e(G)$ for every $e \in E''$.

Fix arbitrary $e \e (w,w') \in E$.
If $w$ is not $u$ or a successor of $u$, i.e, $w \not \in S_{u}(G) \cup \{u\}$, then for every $e' \in E^-_w(G)$ also the start of $e'$ is not a successor of $u$ or $u$, i.e., $\phi_1(e') \not \in S_{u}(G) \cup \{u\}$.
Thus, from Equation~\eqref{eq:axioms:ed:1} and PageRank recursive equation in graph $G$ we have
\begin{equation}
    \label{eq:axioms:ed:2}
    F_{e}(G') = \frac{1}{\deg_{w}^+(G)} \left( a\cdot\sum_{e' \in E^-_w(G)} F_{e'}(G') + b(w)\right) \!\!.
\end{equation}
Since $w \neq u$, we have $\deg_{w}^+(G)=\deg_{w}^+(G')$ and since $w \neq v$, we have $E^-_w(G)=E^-_w(G')$.
Thus, Equation~\eqref{eq:axioms:ed:2} is equivalent to PageRank recursive equation for centrality $F(G')$ and edge $e$.

Now, if $w$ is $u$ or a successor of $u$, then
from Equation~\eqref{eq:axioms:ed:1} and PageRank recursive equation in graph $G''$ we have
\begin{equation*}
    F_{e}(G') = \frac{1}{\deg_{w}^+(G'')} \left( a\cdot\sum_{e' \in E^-_w(G'')} F_{e'}(G') + b''(w)\right) \!\!.
\end{equation*}
Observe that in graph $G''$ outgoing edges of nodes in $S_u(G) \cup \{u\}$ are the same as in $G'$, hence we have
$\deg_{w}^+(G'')=\deg_w^+(G')$.
However, incoming edges from nodes other than $S_u(G) \cup \{u\}$ are not present in $G''$, thus
$E^-_w(G'') = E^-_w(G') \setminus E''$.
Therefore, we get that
\begin{equation}
    \label{eq:axioms:ed:3}
    F_{e}(G') = \frac{1}{\deg_{w}^+(G')} \!\left( a\cdot\!\sum_{e' \in E^-_w(G') \setminus E''}\! F_{e'}(G') + b''(w)\right) \!\!.
\end{equation}
We have put
$b''(w)=b(w) + a \cdot \sum_{e \in E^-_w(G') \cap E''} PR^a_e(G)$.
Thus, from Equation~\eqref{eq:axioms:ed:1} we obtain that this is equivalent to 
$b''(w)=b(w) + a \cdot \sum_{e \in E^-_w(G') \cap E''} F_e(G')$.
In this way, if we substitute it into Equation~\eqref{eq:axioms:ed:3}, we get
\[
    F_{e}(G') = \frac{1}{\deg_{w}^+(G')} \left( a\cdot\sum_{e' \in E^-_w(G')} F_{e'}(G') + b(w)\right) \!\!,
\]
which is PageRank recursive equation for centrality $F(G')$ and edge $e$.
Thus, $F(G')$ satisfies PageRank recursive equation and Edge Deletion follows.

\subsubsection{Edge Multiplication}
Consider edge $e^* \e (u,v) \in E$ such that $E^+_u(G) \subseteq E^-_v(G)$ and
graph $G' = (V ,E',\phi|_{E'},b)$, where $E'=E \setminus E^+_u(G) \cup \{e^*\}$.
Let centrality $F$ be such that
\begin{equation}
    \label{eq:axioms:em:1}
    F_e(G') =
    \begin{cases}
        \deg^+_u(G) \cdot PR^a_{e^*}(G) & \mbox{if } e = e^*,\\
        PR^a_e(G) & \mbox{otherwise.}
    \end{cases}
\end{equation}
We will show that $F(G')$ satisfies PageRank recursive equation (Equation~\eqref{eq:pagerank}) for graph $G'$, 
which will imply that $F(G')=PR^a(G')$ and thus $PR^a_e(G')=PR^a_e(G)$ for every $e \in E' \setminus \{e^*\}$ and $PR^a_{e^*}(G')=\deg^+_u(G) \cdot PR^a_{e^*}(G)$.

Fix arbitrary $e \e (w,w') \in E$.
From Equation~\eqref{eq:pagerank} for PageRank in graph $G$ we get that
\begin{equation}
    \label{eq:axioms:em:2}
    PR^a_{e}(G) = 
        \frac{1}{\deg_{w}^+(G)} \!
        \Bigg( a\cdot \!\sum_{e' \in E^-_w(G) \setminus E^+_u(G)}\! PR^a_{e'}(G) \ + 
    a\cdot\sum_{e' \in E^-_w(G) \cap E^+_u(G)} PR^a_{e'}(G) + b(w)\! \Bigg) \!.
\end{equation}

We will consider four cases:
first, $w \neq u$ and $w \neq v$ (1.1),
next, $w \neq u$, but $w = v$ (1.2),
then $w = u$, but $w \neq v$ (2.1),
and finally $w = u = v$ (2.2).

(1.1) First, let us assume that $w \neq u$.
This means that $e \in E' \setminus \{e^*\}$.
If additionally, $w \neq v$, then we know that $E^-_w(G) \cap E^+_u(G) = \emptyset$ since $E^+_u(G) \subseteq E^-_v(G)$.
From this, we obtain also that $E^-_w(G) \setminus E^+_u(G) = E^-_w(G)=E^-_w(G')$.
Therefore, from Equation~\eqref{eq:axioms:em:1} combined with the fact that $\deg_w^+(G)=\deg_w^+(G')$,
we get that Equation~\eqref{eq:axioms:em:2} is equivalent to PageRank recursive equation for centrality $F(G')$ and edge $e$.

(1.2) If $w = v$,
then from the fact that  $E^+_u(G) \subseteq E^-_v(G)$,
we get $E^-_w(G) \cap E^+_u(G) = E^+_u(G)$.
From Equation~\eqref{eq:pagerank}, each outgoing edge of $u$ has PageRank equal to PageRank of $e^*$, i.e, $PR^a_{e'}(G)=PR^a_{e^*}(G)$ for every $e' \in E^+_u(G)$.
Also, observe that $E^-_w(G) \setminus E^+_u(G) = E^-_w(G') \setminus \{e^*\}$.
Hence, combining this with the fact that $\deg_w^+(G)=\deg_w^+(G')$, Equation~\eqref{eq:axioms:em:2} can be transformed into
\begin{equation*}
    PR^a_{e}(G) = 
        \frac{1}{\deg_{w}^+(G')} \!
        \Bigg( a\cdot \!\sum_{e' \in E^-_w(G') \setminus \{e^*\}}\! PR^a_{e'}(G) \ + 
    a\cdot\deg^+_u(G) \cdot PR^a_{e^*}(G) + b(w)\!\Bigg) \!.
\end{equation*}
And from Equation~\eqref{eq:axioms:em:1} this is equivalent to PageRank recursive equation for centrality $F(G')$ and edge $e$.

(2.1) Now, assume that $w =u$.
Observe that the only outgoing edge of $u$ in $G'$ is $e^*$, thus $e = e^*$.
If $e^*$ is not a self-loop, i.e., $u \neq v$, then $E^-_u(G) \cap E^+_u(G) = \emptyset$ and also $E^-_u(G) \setminus E^+_u(G) = E^-_u(G)=E^-_u(G')$.
Thus, from Equation~\eqref{eq:axioms:em:2} we get
\begin{equation*}
    PR^a_{e^*}(G) = \frac{1}{\deg_{u}^+(G)} \! \Bigg( a\cdot\!\sum_{e' \in E^-_u(G')}\! F_{e'}(G') + b(u) \! \Bigg) \!.
\end{equation*}
By Equation~\eqref{eq:axioms:em:1}, we can transform this equation into
\[
    F^a_{e^*}(G) = a\cdot\!\sum_{e' \in E^-_u(G')}\! F_{e'}(G') + b(u).
\]
Since $\deg^+_u(G')=1$, this equation is PageRank recursive equation for centrality $F(G')$ and edge $e$.

(2.2) Finally, if $e^*$ is a self-loop, then each outgoing edge of $u$ has PageRank equal to PageRank of $e^*$, i.e, $PR^a_{e}(G)=PR^a_{e^*}(G)$ for every $e \in E^+_u(G)$.
Also, observe that $E^-_u(G) \setminus E^+_u(G) = E^-_u(G') \setminus \{e^*\}$.
Hence, Equation~\eqref{eq:axioms:em:2} can be transformed into
\begin{equation*}
    PR^a_{e^*}(G) = 
        \frac{1}{\deg_{u}^+(G)} \!
        \Bigg( a\cdot \!\sum_{e' \in E^-_u(G') \setminus \{e^*\}}\! PR^a_{e'}(G) \ + 
    a\cdot\deg^+_u(G) \cdot PR^a_{e^*}(G) + b(u)\!\Bigg) \!.
\end{equation*}
Finally, by Equation~\eqref{eq:axioms:em:1} we can transform this equation into
\[
    F^a_{e^*}(G) = a\cdot\!\sum_{e' \in E^-_u(G')}\! F_{e'}(G') + b(u).
\]
Since $\deg^+_u(G')=1$, this equation is PageRank recursive equation for centrality $F(G')$ and edge $e$.
Thus, $F(G')$ satisfies PageRank recursive equation and Edge Multiplication follows.

\subsubsection{Edge Swap}
Consider edges, $e_1 \ee (u_1,v_1), e_2 \ee (u_2,v_2) \in E$ such that $PR^a_{e_1}(G)=PR^a_{e_2}(G)$ and let us denote a graph $G'=(V,E,\phi[e_1 \rightarrow (u_1,v_2), e_2 \rightarrow (u_2,v_1)],b)$.
Let centrality $F$ be such that $F_e(G')=PR^a_e(G)$ for every $e \in E$.
We will show that $F(G')$ satisfies PageRank recursive equation (Equation~\eqref{eq:pagerank}) for graph $G'$, 
which will imply that $F(G')=PR^a(G')$ and thus $PR^a_e(G')=PR^a_e(G)$ for every $e \in E$.

Fix arbitrary $e \e (w,w') \in E$.
From Equation~\eqref{eq:pagerank} for PageRank in graph $G$ and the fact that $F_e(G')=PR^a_e(G)$ for every $e \in E$ we get
\begin{equation}
    \label{eq:axioms:es:1}
    F^a_{e}(G') = 
        \frac{1}{\deg_{w}^+(G)} \!
        \Bigg( a\cdot \!\sum_{e' \in E^-_w(G) \setminus \{e_1,e_2\}}\! F_{e'}(G') \ + 
    a\cdot\sum_{e' \in E^-_w(G) \cap \{e_1,e_2\}} F_{e'}(G') + b(w)\! \Bigg) \!.
\end{equation}
Out-degree of every node stays the same in $G$ and $G'$, thus $\deg_w^+(G)=\deg_w^+(G')$.
Edges that are not $e_1$ or $e_2$ are unaffected, hence $E^-_w(G) \setminus \{e_1,e_2\}=E^-_w(G') \setminus \{e_1,e_2\}$.
Finally, observe that $PR^a_{e_1}(G)=PR^a_{e_2}(G)$ implies that $F_{e_1}(G')=F_{e_2}(G')$.
Thus, since number of edges in $\{e_1,e_2\}$ that ends in $w$ remains the same, we have
\begin{equation*}
    \sum_{e' \in E^-_w(G) \cap \{e_1,e_2\}} \! F_{e'}(G') =
    \sum_{e' \in E^-_w(G') \cap \{e_1,e_2\}} F_{e'}(G').
\end{equation*}
Combining all three facts with Equation~\eqref{eq:axioms:es:1}, we get
\begin{equation*}
    F_{e}(G') = 
        \frac{1}{\deg_{w}^+(G')} \!
        \Bigg( a\cdot \!\sum_{e' \in E^-_w(G') \setminus \{e_1,e_2\}}\! F_{e'}(G') \ + 
    a\cdot\sum_{e' \in E^-_w(G') \cap \{e_1,e_2\}} F_{e'}(G') + b(w)\! \Bigg) \!.
\end{equation*}
This equation is PageRank recursive equation for centrality $F(G')$ and edge $e$.
Thus, $F(G')$ satisfies PageRank recursive equation and Edge Swap follows.

\subsubsection{Node Redirect}
Consider out-twins $u,w \in V$ with the bijection $\psi : E^+_w(G) \rightarrow E^+_u(G)$ and let us denote the graph $G' = (V',E',\phi',b') = R_{u \rightarrow w}(G)$.
Let centrality $F$ be such that for every $e \in E'$
\begin{equation}
    \label{eq:axioms:nr:1}
    F_e(G') =
    \begin{cases}
        PR^a_{e}(G) + PR^a_{\psi(e)}(G) & \mbox{if } e \in E^+_w(G),\\
        PR^a_e(G) & \mbox{otherwise.}
    \end{cases}
\end{equation}
We will show that $F(G')$ satisfies PageRank recursive equation (Equation~\eqref{eq:pagerank}) for graph $G'$, 
which will imply that $F(G')=PR^a(G')$ and thus $PR^a_e(G')=PR^a_e(G)$ for every $e \!\in\! E' \setminus E^+_w(G')$ and $PR^a_{e}(G')=PR^a_{e}(G) + PR^a_{\psi(e)}(G)$ for every $e \in E^+_w(G')$.

Fix $e \e (v,v') \in E'$.
Observe that out-degree of each node in $V'$ is the same in both $G$ and $G'$.
Thus, from  Equation~\eqref{eq:pagerank} for PageRank in graph $G$ we get
\begin{equation}
    \label{eq:axioms:nr:2}
    PR^a_{e}(G) \! = \! 
        \frac{1}{\deg_{v}^+(G')} \!
        \Bigg(
            a\cdot \!\sum_{e' \in E^-_v(G) \setminus E^+_{u,w}(G)}\! PR^a_{e'}(G) \ + 
            a\cdot\sum_{e' \in E^-_v(G) \cap E^+_{u,w}(G)} PR^a_{e'}(G) + 
            b(v)\! 
        \Bigg) \!,
\end{equation}
where $E^+_{u,w}(G) = E^+_u(G) \cup E^+_w(G)$.
Now, since $u$ and $w$ are out-twins with bijection $\psi$, it holds that
\begin{equation*}
    \sum_{e' \in E^-_v(G) \cap E^+_{u,w}(G)} PR^a_{e'}(G) = 
    \sum_{e' \in E^-_v(G) \cap E^+_{w}(G)} PR^a_{e'}(G) + PR^a_{\psi(e')}(G).
\end{equation*}
Combining this with Equations~\eqref{eq:axioms:nr:1} and~\eqref{eq:axioms:nr:2} we obtain
\begin{equation}
    \label{eq:axioms:nr:3}
    PR^a_{e}(G) \! = \! 
        \frac{1}{\deg_{v}^+(G')} \!
        \Bigg(
            a\cdot \!\sum_{e' \in E^-_v(G) \setminus E^+_{u,w}(G)}\! F_{e'}(G') \ + 
            a\cdot\sum_{e' \in E^-_v(G) \cap E^+_{w}(G)} F_{e'}(G') + 
            b(v)\! 
        \Bigg) \!.
\end{equation}

If $v \neq w$, then $E^-_v\!( G) \setminus E^+_{u,w}\!( G) \! = \! E^-_v\! ( G') \setminus E^+_{w}\! ( G')$ 
and $E^-_v(G) \cap E^+_{w}(G) = E^-_v(G') \cap E^+_{w}(G')$.
Also, we have $b'(v)=b(v)$.
Moreover, from Equation~\eqref{eq:axioms:nr:1} we get that $F_e(G') = PR^a_e(G)$.
Thus, Equation~\eqref{eq:axioms:nr:3} is equivalent to PageRank recursive equation for $F(G')$ and edge $e$.

If, on the other hand, $v = w$, then let us sidewise add Equation~\eqref{eq:axioms:nr:3} for $e$ with Equation~\eqref{eq:axioms:nr:3} for $\psi(e)$.
Since from Equation~\eqref{eq:axioms:nr:1} we have $F_e(G') = PR^a_e(G) + PR^a_{\psi(e)}(G)$, we obtain
\begin{equation}
    \label{eq:axioms:nr:4}
    F_{e}(G') \! = \! 
        \frac{1}{\deg_{v}^+(G')} \!
        \Bigg(
            a\cdot \!\sum_{e' \in E^-_{u,w}(G) \setminus E^+_{u,w}(G)}\! F_{e'}(G') \ +  \\
            a\cdot\sum_{e' \in E^-_{u,w}(G) \cap E^+_{w}(G)} F_{e'}(G') + 
            b(u) + b(w)\! 
        \Bigg) \!,
\end{equation}
where we denote $E^-_{u,w}(G)=E^-_{u}(G) \cup E^-_{w}(G)$.
Observe that $E^-_{u,w}(G) \setminus E^+_{u,w}(G) = E^-_{w}(G') \setminus E^+_{w}(G')$ and similarily
$E^-_{u,w}(G) \cap E^+_{u,w}(G) = E^-_{w}(G') \cap E^+_{w}(G')$.
Also, $b'(w)= b(u) + b(w)$.
Thus, Equation~\eqref{eq:axioms:nr:4} is equivalent to PageRank recursive equation for centrality $F(G')$ and edge $e$.
Therefore, $F(G')$ satisfies PageRank recursive equation and Node Redirect follows.

\subsubsection{Baseline}
Consider arbitrary isolated edge $e \e (u,v) \in E$.
Observe that $PR^a_e(G) = b(u)$ directly from Equation~\eqref{eq:pagerank}.
\end{proof}

\subsection{Axioms Imply Edge PageRank}
Now, let us focus on proving that an arbitrary centrality measure $F$ satisfying our axioms is indeed Edge PageRank.
To this end, we adopt the following structure of the proof:

\begin{itemize}

\item First, we prove simple properties of $F$ (Proposition~\ref{prop:loc-sw-se}). Specifically, we show that the centrality of an edge depends only on its connected component (\emph{Locality}), does not depend on weights of sinks (\emph{Sink Weight}), and if the edge is a source edge, then it is equal to the weight of its start (\emph{Source Edge}).

\item Then, we analyze simple graphs in which all edges are incident to one node, $v$. 
\begin{itemize}
\item We begin with graphs in which $v$ has one incoming edge (from $u$) and one outgoing edge (to $w$).
First, we show that there exists a constant $a_F$ such that the centrality of the outgoing edge of $v$ equals $a_F$ times $b(u)$ plus $b(v)$ (Lemma~\ref{lemma:two_path_weighted}). Then, we prove that $a_F \in [0,1)$ (Lemma~\ref{lemma:a_f}).

\item Furthermore, we consider graphs in which node $v$ has $k$ outgoing and zero (Lemma~\ref{lemma:star}) or one (Lemma~\ref{lemma:star_edge}) incoming edges.
\end{itemize}    
\item Finally, we prove that for arbitrary graph, $F$ is equal to Edge PageRank with the decay factor $a_F$ (Lemma~\ref{lemma:any_graph}).

\end{itemize}


We begin with a proposition that captures simple properties of an edge centrality measure implied by our axioms.
\begin{apptheorem}{Proposition}{\ref{prop:loc-sw-se}}
If edge centrality measure $F$ satisfies Node Deletion, Edge Deletion, Node Redirect, and Baseline, then for every graph $G = (V, E,\phi, b)$  and $e \in E$ it holds that
\begin{itemize}
    \item[(a)] (Locality) $F_e(G)=F_e(G + G')$ for every $G'$ disjoint with $G$,
    \item[(b)] (Sink Weight) $F_e(G)=F_e(V,E,\phi,b[w \rightarrow 0])$ for every sink $w \in V$,
    \item[(c)] (Source Edge) $F_e(G)=b(u)$ if $e \e (u,v)\in E$ is a source edge.
\end{itemize}
\end{apptheorem}
\begin{proof}
    For (a),
    fix arbitrary graph $G'=(V',E',\phi',b')$ such that $V \cap V' = \emptyset$ and $E \cap E' = \emptyset$.
    Also, consider $G''$ that is $G'$ without any edges, i.e., $G'' = (V',\emptyset, \phi'|_{\emptyset},b')$.
    Since in graph $G+G'$ node $\phi_1(e)$ is not a successor of nodes from $V'$,
    by Edge Deletion we obtain that
    $F_e(G + G')=F_e(G + G'')$.
    Observe that in graph $G + G''$ all nodes from $V'$ are isolated.
    Hence, from Node Deletion we have $F_e(G + G'') = F_e(G)$.
    Combining both equations we get the thesis.
    
    For (b), consider arbitrary sink $w \in V$ and graph $G''$ that is $G$ in which the weight of sink $w$ is set to zero, i.e., let
    $G''=(V,E,\phi,b[w \rightarrow 0])$.
    Also, let us add to $G''$ an isolated node, $w' \not \in V$, with the weight of the original node $w$, i.e., let
    $G'=G'' + (\{w'\},\emptyset,\phi|_\emptyset,[b(w)])$.
    From Locality (a), we get that $F_e(G'')=F_e(G')$.
    On the other hand, observe that graph $G$ can be obtained from $G'$ by redirecting $w'$ into $w$, i.e.,
    $G = R_{w' \rightarrow w}(G')$.
    Since both nodes are sinks, they are out-twins.
    Thus, from Node Redirect $F_e(G') = F_e(G)$.
    Combining both equations we get the thesis.
    
    Finally, for (c), consider graph $G$ with all edges removed except for $e$, i.e., $G'=(V,e,\phi|_{e},b)$.
    Observe that since $e$ is a source edge, its start, $u$, is not a successor of any node.
    Thus, from Edge Deletion, we get $F_e(G) = F_e(G')$.
    On the other hand, in graph $G'$ edge $e$ is isolated.
    Thus, from Baseline we get $F_e(G')=b(u)$.
    Combining both equations we get the thesis.
\end{proof}

In the following lemma, we introduce constant $a_F$ that, as we will prove later on, is decay factor of the Edge PageRank to which $F$ is equal to.

\begin{apptheorem}{Lemma}{\ref{lemma:two_path_weighted}}
If edge centrality measure $F$ satisfies Node Deletion, Edge Deletion, Node Redirect, and Baseline, then there exists a constant, $a_F \in \mathbb{R}$, such that for every graph
$G \!=\! (\{u, v, w\},  \{e_1 \ee (u,v), e_2 \ee (v, w)  \},  [x, y, 0])$ it holds that
\[
    F_{e_1} (G) = x \quad \mbox{and} \quad F_{e_2} (G) = a_F \cdot x + y.
\]
\end{apptheorem}

\begin{proof}[Proof]
Note that first equation follows directly from Source Edge (Proposition~\ref{prop:loc-sw-se}c).
Thus, we will focus on proving that $F_{e_2} (G) = a_F \cdot x + y$.
First, assume $y=0$ and consider graph
$G' \!\!=\! ( \{u'\!, v'\!, w'\}, \! \{e'_1 \ee (u' \! , v'), e'_2 \ee (v'\!, w') \! \}, \! [x'\!, 0, 0])$
such that $u',v',w' \not \in \{u,v,w\}$
(see Figure~\ref{fig:app:two_path_weighted:1}).
By Locality (Proposition~\ref{prop:loc-sw-se}a), we know that in $G+G'$ the centrality of all edges is the same as in $G$ or $G'$.
Specifically,
\begin{equation}
    \label{eq:lemma:2-path:1}
    F_{e_2}(G + G') = F_{e_2}(G)
    \quad \mbox{and} \quad
    F_{e'_2}(G + G') = F_{e'_2}(G').
\end{equation}

\begin{figure*}[t]
\centering
\begin{tikzpicture}[x=4cm,y=4cm] 
  \tikzset{     
    e4c node/.style={circle,draw,minimum size=0.6cm,inner sep=0},
    e4c edge/.style={above,font=\footnotesize},
    e4cn edge/.style={above,font=\footnotesize, pos=0.6}
  }
  
  \def\x{0};
  \def\y{0};
  \node (G) at (\x+0.3, \y+0.47) {$G$}; 
  \node (G_) at (\x+0.31, \y+0.13) {$G'$}; 
  \node[e4c node] (1) at (\x+0, \y+0.3) {$u$}; 
  \node[e4c node] (2) at (\x+0.3, \y+0.3) {$v$}; 
  \node[e4c node] (3) at (\x+0.6, \y+0.3) {$w$}; 
  \node[e4c node] (4) at (\x+0, \y+0) {$u'$}; 
  \node[e4c node] (5) at (\x+0.3, \y+0) {$v'$}; 
  \node[e4c node] (6) at (\x+0.6, \y+0) {$w'$}; 

  \path[->,draw,thick]
  (1) edge[e4c edge] node  {$e_1$}  (2) 
  (2) edge[e4c edge] node  {$e_2$}  (3)
  (4) edge[e4c edge] node  {$e'_1$} (5)
  (5) edge[e4c edge] node  {$e'_2$} (6)
  ;

  \def\x{0.9};
  \node (G) at (\x+0.3, \y+0.47) {$R_{w' \rightarrow w}(G+G')$};
  \node[e4c node] (1) at (\x+0, \y+0.3) {$u$}; 
  \node[e4c node] (2) at (\x+0.3, \y+0.3) {$v$}; 
  \node[e4c node] (3) at (\x+0.6, \y+0.3) {$w$}; 
  \node[e4c node] (4) at (\x+0, \y+0) {$u'$}; 
  \node[e4c node] (5) at (\x+0.3, \y+0) {$v'$}; 

  \path[->,draw,thick]
  (1) edge[e4c edge] node  {$e_1$}  (2) 
  (2) edge[e4c edge] node  {$e_2$}  (3)
  (4) edge[e4c edge] node  {$e'_1$} (5)
  (5) edge[e4c edge] node[pos =0.3]  {$e'_2$} (3)
  ;
  
  \def\x{1.8};
  \node (G) at (\x+0.3, \y+0.47) {$R_{v' \rightarrow v}(R_{w' \rightarrow w}(G+G'))$};
  \node[e4c node] (1) at (\x+0, \y+0.3) {$u$}; 
  \node[e4c node] (2) at (\x+0.3, \y+0.3) {$v$}; 
  \node[e4c node] (3) at (\x+0.6, \y+0.3) {$w$}; 
  \node[e4c node] (4) at (\x+0, \y+0) {$u'$};

  \path[->,draw,thick]
  (1) edge[e4c edge] node  {$e_1$}  (2) 
  (2) edge[e4c edge] node  {$e_2$}  (3)
  (4) edge[e4c edge] node[pos =0.3]  {$e'_1$} (2)
  ;
  
  \def\x{2.7};
  \node (G) at (\x+0.3, \y+0.47) {$G''$};
  \node[e4c node] (1) at (\x+0, \y+0.3) {$u$}; 
  \node[e4c node] (2) at (\x+0.3, \y+0.3) {$v$}; 
  \node[e4c node] (3) at (\x+0.6, \y+0.3) {$w$};

  \path[->,draw,thick]
  (1) edge[e4c edge] node  {$e_1$}  (2) 
  (2) edge[e4c edge] node  {$e_2$}  (3)
  ;
  


\end{tikzpicture}
\caption{An illustration to the first part of the proof of Lemma~\ref{lemma:two_path_weighted}.}
\label{fig:app:two_path_weighted:1}
\end{figure*}
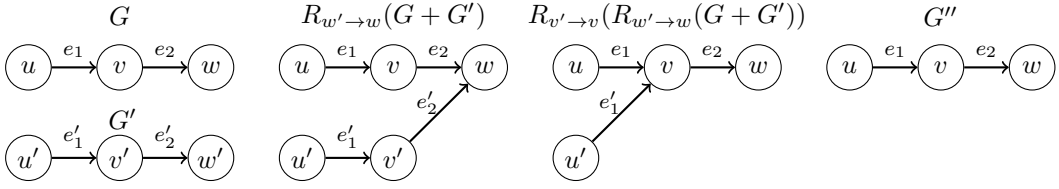

Now, in graph $G+G'$, we sequentially redirect nodes $w',v',u'$ into $w,v,u$ respectively, to again obtain graph $G$ but with different weight of node $u$.
Formally, $G'' = R_{w' \rightarrow w}(R_{v' \rightarrow v}(R_{u' \rightarrow u}(G)))$.
It is easy to check, that in fact
$G'' \! \!=\! ( \{u,\! v,\! w\}, \! \{e_1 \ee (u,v), e_2 \ee (v, w) \! \}, \! [x + x'\! , 0, 0])$.
Moreover, observe that in each redirection, we  redirected a node into its out-twin.
Thus, from Node Redirect and Equation~\eqref{eq:lemma:2-path:1}, we obtain that 
\begin{multline}
    \label{eq:lemma:2-path:2}
    F_{e_2}(\!\{u,\! v,\! w\}, \! \{e_1 \ee (u,v), e_2 \ee (v, w) \! \}, \! [x + x'\! , 0, 0])= \\
        F_{e_2}(\!\{u, v, w\},  \{e_1 \ee (u,v), e_2 \ee (v, w) \! \},  [x, 0, 0])) +
        F_{e'_2}(\!\{u'\!, v'\!, w'\}, \! \{e'_1 \ee (u' \! , v'), e'_2 \ee (v'\!, w') \! \}, \! [x'\!, 0, 0]).
\end{multline}
This equation implies the following:
\begin{itemize}
    \item[i)]
    $F_{e'_2}(\{u'\!, v'\!, w'\}, \! \{e'_1 \e (u' \! , v'), e'_2 \e (v'\!, w')  \},  [0, 0, 0]) = 0$,\\
    from Equation~\eqref{eq:lemma:2-path:2} for $x'=0$ and arbitrary $x \in \mathbb{R}_{\ge 0}$,
    \item[ii)]
    $F_{e_2}( \{u, v, w\}, \{e_1 \e (u , v), e_2 \e (v, w) \},  [0, 0, 0]) = 0$,\\
    from \emph{i)} and the arbitrariness of the choice of nodes and edges,
    \item[iii)]
    $F_{e_2}( \{u, v, w\}, \{e_1 \e (u , v), e_2 \e (v, w) \},  [c, 0, 0]) =$\\
    $F_{e'_2}(\{u'\!, v'\!, w'\}, \! \{e'_1 \e (u' \! , v'), e'_2 \e (v'\!, w')  \},  [c, 0, 0])$\\
    for every $c \in \mathbb{R}_{\ge 0}$, from \emph{ii)} and Equation~\eqref{eq:lemma:2-path:2} for $x=0$ and $x' = c$, and finally
    \item[iv)]
    $F_{e_2}( \{u, v, w\}, \{e_1 \e (u , v), e_2 \e (v, w) \},  [c, 0, 0]) +$\\
    $F_{e_2}( \{u, v, w\}, \{e_1 \e (u , v), e_2 \e (v, w) \},  [c', 0, 0]) =$\\
    $F_{e_2}( \{u, v, w\}, \{e_1 \e (u , v), e_2 \e (v, w) \},  [c + c', 0, 0]),$\\
    for every $c,c' \in \mathbb{R}_{\ge 0}$, from \emph{iii)} and Equation~\eqref{eq:lemma:2-path:2} for $x=c$ and $x' = c'$.
\end{itemize}
Observe that from \emph{iii)} we get that centrality $F_{e_2}(G)$ depends solely on the value of $x$, i.e.,
there exists a function $f : \mathbb{R}_{\ge 0} \rightarrow \mathbb{R}_{\ge 0}$, such that $F_{e_2}(G)=f(x)$.
Also, from \emph{iv)} we know that this function is additive.
Moreover, centrality is always non-negative, thus $f$ is non-negative.
Therefore, since additive and non-negative function is always linear~\cite{Cauchy:1821}, we get that $f$ is of the form $f(x)=a_F \cdot x$.

  
  





It remains to consider the case in which $y$ is not necessarily equal to 0.
To this end, consider the following graph:
$G^* \! \!=\! ( \{u,\! v,\! v' \!,\! w\}, \! \{e_1 \ee (u,\! v), e_2 \ee (v,\! w), e_3 \ee (v' \!,\! w)  \}, \! [x, 0, y, 0])$ (see Figure~\ref{fig:app:two_path_weighted:2}).
Observe that from Source Edge (Proposition~\ref{prop:loc-sw-se}c) we get that 
\begin{equation}
    \label{eq:lemma:2-path:3}
    F_{e_3}(G^*)=y.
\end{equation}

Now, let us remove edge $e_3$ and then node $v'$.
In result, we obtain graph
$G^\star \! \!=\! ( \{u,\! v, \! w\}, \! \{e_1 \ee (u,\! v), e_2 \ee (v,\! w)  \}, \! [x, 0, 0])$.
Since $v$ is not a successor of $v'$ and after removal of $e_3$ node $v'$ is isolated, from Edge Deletion and Node Deletion, we get that $F_{e_2}(G^*)=F_{e_2}(G^\star)$.
%
On the other hand, observe that $G^\star$ is a type of graph that we have considered in the first part of the proof.
Thus, we get that $F_{e_2}(G^\star)=a_F \cdot x$.
Therefore, 
\begin{equation}
    \label{eq:lemma:2-path:5}
    F_{e_2}(G^*)=a_F \cdot x.
\end{equation}

Finally, observe that by redirecting node $v'$ into $v$ in graph $G^*$ we obtain the original graph, $G$.
Since $v'$ and $v$ are out-twins in $G^*$ (both have one edge to node $w$), from Node Redirect we get that
$F_{e_2}(G)=F_{e_2}(G^*) + F_{e_3}(G^*)$.
Combining this with Equation~\eqref{eq:lemma:2-path:3} and~\eqref{eq:lemma:2-path:5} yields the thesis.
\end{proof}

Now, let us show the bounds for constant $a_F$.


  



\begin{figure}[t]
\centering
\begin{minipage}{.5\textwidth}
    \centering
    \begin{tikzpicture}[x=4cm,y=4cm] 
      \tikzset{     
        e4c node/.style={circle,draw,minimum size=0.6cm,inner sep=0},
        e4c edge/.style={above,font=\footnotesize},
        e4cn edge/.style={above,font=\footnotesize, pos=0.6}
      }
      
      \def\x{0};
      \def\y{0};
      
      \node (G) at (\x+0.2, \y+0.45) {$G^*$}; 
      \node[e4c node] (1) at (\x+0, \y+0.3) {$u$}; 
      \node[e4c node] (2) at (\x+0.25, \y+0.3) {$v$}; 
      \node[e4c node] (3) at (\x+0.5, \y+0.3) {$w$}; 
      \node[e4c node] (5) at (\x+0.25, \y+0) {$v'$};
    
      \path[->,draw,thick]
      (1) edge[e4c edge] node {$e_1$} (2) 
      (2) edge[e4c edge] node  {$e_2$}  (3)
      (5) edge[e4cn edge, pos=0.6, below] node [yshift=-2] {$e_3$}  (3)
      ;
    
      \def\x{0.8};
      \node (G) at (\x+0.2, \y+0.45) {$G^\star$}; 
      \node[e4c node] (1) at (\x+0, \y+0.3) {$u$}; 
      \node[e4c node] (2) at (\x+0.25, \y+0.3) {$v$}; 
      \node[e4c node] (3) at (\x+0.5, \y+0.3) {$w$}; 
    
      \path[->,draw,thick]
      (1) edge[e4c edge] node {$e_1$} (2) 
      (2) edge[e4c edge] node  {$e_2$}  (3)
      ;
    
    \end{tikzpicture}
    \captionof{figure}{An illustration to the second part of the\\ proof of Lemma~\ref{lemma:two_path_weighted}.}
    \label{fig:app:two_path_weighted:2}
\end{minipage}%
\begin{minipage}{.5\textwidth}
    \centering
    \begin{tikzpicture}[x=4cm,y=4cm] 
      \tikzset{     
        e4c node/.style={circle,draw,minimum size=0.6cm,inner sep=0}, 
        e4c edge/.style={sloped,above,font=\footnotesize},
        e4cn edge/.style={above,pos=0.6,font=\footnotesize}
      }
      \def\x{0}
      \def\h{0.3}
      \node (G) at (\x-0.1, \h + 0.05) {$G$}; 
      \node[e4c node] (1) at (\x+0.2, \h) {$u$}; 
      \node[e4c node] (2) at (\x+0.00, 0) {$w$}; 
      \node[e4c node] (3) at (\x+0.40, 0) {$v$};

      \path[->,draw,thick]
      (3) edge[e4c edge, loop above] node  {$e_2$} (3)
      (1) edge[e4cn edge]  node  [xshift=-3]  {$e_1$} (2)
      ;
      
      \def\x{0.8}
      \node (G) at (\x-0.1, \h + 0.05) {$G'$}; 
      \node[e4c node] (1) at (\x+0.2, \h) {$u$}; 
      \node[e4c node] (2) at (\x+0.00, 0) {$w$}; 
      \node[e4c node] (3) at (\x+0.40, 0) {$v$};

      \path[->,draw,thick]
      (3) edge[e4c edge] node {$e_2$} (2)
      (1) edge[e4cn edge]  node  [xshift=3] {$e_1$} (3)
      ;
    
    \end{tikzpicture}
    \captionof{figure}{An illustration to the proof of Lemma~\ref{lemma:a_f}.}
    \label{fig:app:a_f}
\end{minipage}
\end{figure}

\begin{apptheorem}{Lemma}{\ref{lemma:a_f}}
If edge centrality measure $F$ satisfies Node Deletion, Edge Deletion, Edge Swap, Node Redirect, and Baseline, then $a_F \in [0,1)$.
\end{apptheorem}
\begin{proof}[Proof]
First, to show that $a_F \ge 0$ let us consider graph
$G \!=\! (\{u, v, w\},  \{e_1 \ee (u,v), e_2 \ee (v, w)  \},  [0, y, 0])$
for any $y \in \mathbb{R}_{>0}$.
From Lemma~\ref{lemma:two_path_weighted} we get that $F_{e_2}(G)=a_F \cdot y$.
Since centrality is non-negative and $y>0$ we get that $a_F$ is also non-negative.

Thus, in the remainder of the proof, we focus on showing that $a_F<1$.
To this end, let us consider graph
$G=(\{u,v,w\},\{e_1 \ee (u,w), e_2 \ee (v,v) \},[x,1,0])$ (see Figure~\ref{fig:app:a_f}).
From Locality (Proposition~\ref{prop:loc-sw-se}a), we get that 
\[
    F_{e_2}(G) = F_{e_2}(\{v\},\{e_2 \e (v,v) \},[1]). 
\]
Since this value does not depend on $x$, we can take $x$ that is equal to this value, i.e., $x = F_{e_2}(G)$.
Then, from Source Edge (Proposition~\ref{prop:loc-sw-se}c) we obtain $F_{e_1}(G)= x = F_{e_2}(G)$.
Hence, by Edge Swap, exchanging the ends of these two edges does not affect their centralities.
Formally, if we take graph
$G'=(\{u,v,w\},\{e_1 \ee (u,v), e_2 \ee (v,w) \},[x,1,0])$, then from Edge Swap we have
\begin{equation}
    \label{eq:lemma:a_f:1}
    F_{e_2}(G')=F_{e_2}(G).
\end{equation}
On the other hand, from Lemma~\ref{lemma:two_path_weighted} for graph $G'$, we have
$F_{e_2}(G')=1 + a_F \cdot x$.
Since we took $x=F_{e_2}(G)$, this means that
\begin{equation}
    \label{eq:lemma:a_f:2}
    F_{e_2}(G') = 1 + a_F \cdot F_{e_2}(G).
\end{equation}
Now, combining Equation~\eqref{eq:lemma:a_f:1} with Equation~\eqref{eq:lemma:a_f:2} we obtain that
\(
    F_{e_2}(G) = 1 + a_F \cdot F_{e_2}(G),
\)
which implies that $a_F \neq 1$.
This means that we can further transform this equation into
\(
    F_{e_2}(G) = 1/(1-a_F).
\)
Since centrality is non-negative, we get $a_F<1$.
\end{proof}

In Lemmas~\ref{lemma:star} and \ref{lemma:star_edge} we consider graphs in which $v$ has multiple outgoing edges.
First, we assume that it does not have any incoming edge.

\begin{apptheorem}{Lemma}{\ref{lemma:star}}
If edge centrality measure $F$ satisfies Node Deletion, Edge Deletion, Edge Multiplication, Node Redirect, and Baseline, then for every $k\in \mathbb{N}$ and graph 
$G \!=\! (\!\{v, w_1,\dots, w_k\},\! \{e_1 \ee (v,w_1), \dots, e_k \ee (v,w_k) \!\}, \! [x, 0, \dots, 0])$,
it holds that
\[
    F_{e_i}(G) = x / k \quad \mbox{for every } i\in \{1,\dots, k\}.
\]
\end{apptheorem}

\begin{proof}
Let us fix arbitrary $i \in \{ 1, \dots, k\}$.
Since nodes $w_1,\dots,w_k$ are all sinks, they are also out-twins.
Hence, from Node Redirect, we know that sequentially redirecting nodes $w_2,\dots,w_k$ into node $w_1$ preserves the centrality of edges $e_1,\dots,e_k$.
Therefore, in the obtained graph,
\(
    G' \! = ( \{v, w_1\},\! \{e_1 \ee (v,w_1), \dots, e_k \ee (v,w_1) \}, \! [x, 0]),
\)
we have $F_{e_i}(G) = F_{e_i}(G')$.
See Figure~\ref{fig:app:star} for an illustration.

Next, consider graph
\(
    G_i = ( \{v, w_1\},\!\{e_i \e (v,w_1) \}, \! [x, 0]).
\)
Observe that $G_i$ can be obtained from $G'$ by removing all edges but $e_i$.
Thus, from Edge Multiplication, we get $F_{e_i}(G')=F_{e_i}(G_i) / k$.
On the other hand, from Baseline, $F_{e_i}(G_i)=x$.
Combining all equations, we obtain that
\[
    F_{e_i}(G) = F_{e_i}(G') = \frac{F_{e_i}(G_i)}{k} = \frac{x}{k}.
\]
\end{proof}

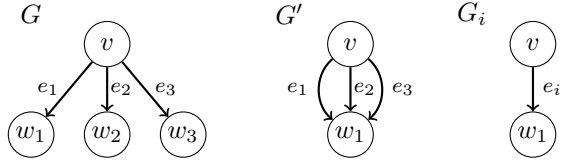
\begin{figure}[t]
\centering
\begin{tikzpicture}[x=4cm,y=4cm] 
  \tikzset{     
    e4c node/.style={circle,draw,minimum size=0.6cm,inner sep=0},
    e4c edge/.style={right,font=\footnotesize},
    e4cn edge/.style={above,font=\footnotesize}
  }
  
  \def\xdist{0.25}
  \def\ydist{0.3}
  \def\b{20} 
  
  \def\x{0};
  \def\y{0};
  \node (G) at (\x, \y+\ydist + 0.1) {$G$};
  \node[e4c node] (2) at (\x+\xdist, \y+\ydist) {$v$};
  \node[e4c node] (3) at (\x+0, \y) {$w_1$}; 
  \node[e4c node] (4) at (\x+\xdist, \y) {$w_2$};
  \node[e4c node] (5) at (\x+2*\xdist, \y) {$w_3$}; 

  \path[->,draw,thick]
  (2) edge[e4c edge] node [left]  {$e_1$}  (3)
  (2) edge[e4c edge] node  [xshift=-2.5]  {$e_2$}  (4)
  (2) edge[e4c edge] node  {$e_3$} (5)
  ;

  \def\x{0.8};
  \node (G) at (\x + \xdist - 0.2, \y+\ydist + 0.1) {$G'$};
  \node[e4c node] (2) at (\x+\xdist, \y+\ydist) {$v$};
  \node[e4c node] (3) at (\x+\xdist, \y) {$w_1$};
  
  \path[->,draw,thick]
  (2) edge[e4c edge, bend right=50]  node [left]  {$e_1$} (3)
  (2) edge[e4c edge] node  [xshift=-2.5]  {$e_2$} (3)
  (2) edge[e4c edge, bend left=50] node  {$e_3$} (3)
  ;

  \def\x{1.4};
  \node (G) at (\x + \xdist - 0.2, \y+\ydist + 0.1) {$G_i$};
  \node[e4c node] (2) at (\x+\xdist, \y+\ydist) {$v$};
  \node[e4c node] (3) at (\x+\xdist, \y) {$w_1$};
  
  \path[->,draw,thick]
  (2) edge[e4c edge]  node  {$e_i$}  (3)
  ;
\end{tikzpicture}
\caption{An illustration to the proof of Lemma~\ref{lemma:star} for $k=3$.}
\label{fig:app:star}
\end{figure}

Now, we assume that $v$ has exactly one incoming edge and multiple outgoing edges.

\begin{apptheorem}{Lemma}{\ref{lemma:star_edge}}
If edge centrality measure $F$ satisfies Node Deletion, Edge Deletion, Edge Multiplication, Node Redirect, and Baseline, then for every $k\in \mathbb{N}$ and graph 
\[\begin{split}
G = (&\{u, v, w_1,\dots w_k\}, \\
&\{e \e (u,v), e_1 \e (v,w_1), \dots, e_k \e (v,w_k)\}, \\
&[x, y, 0, \dots, 0]), 
\end{split}\]
it holds that $F_e(G) = x$ and $F_{e_i}(G) = (a_F \cdot x + y) / k$ for every $i\in \{1,\dots, k\}$.
\end{apptheorem}
\begin{proof}[Proof]
The proof is analogous to the proof of Lemma~\ref{lemma:star} with only two differences:
first, in every graph node $v$ has one incoming edge from $u$;
second, the final equation for graphs $G_i$ is of the form $F_{e_i}(G_i)= a_F \cdot x + y$.

Let us fix arbitrary $i \in \{ 1, \dots, k\}$.
Since nodes $w_1,\dots,w_k$ are all sinks, they are also out-twins.
Hence, from Node Redirect, we know that sequentially redirecting nodes $w_2,\dots,w_k$ into node $w_1$ preserves the centrality of edges $e_1,\dots,e_k$.
Therefore, in the obtained graph,
\(
    G' \! = ( \{u, v, w_1\},\! \{e \ee (u,v), e_1 \ee (v,w_1), \dots, e_k \ee (v,w_1) \}, \! [x,y, 0]),
\)
we have $F_{e_i}(G) = F_{e_i}(G')$.

Let
\(
    G_i = ( \{u, v, w_1\},\!\{e \ee (u,v), e_i \ee (v,w_1) \}, [x, y, 0]).
\)
Observe that $G_i$ can be obtained from $G'$ by removing all edges but $e_i$.
Thus, from Edge Multiplication, we get $F_{e_i}(G')=F_{e_i}(G_i) / k$.
On the other hand, from Lemma~\ref{lemma:two_path_weighted} we get, $F_{e_i}(G_i)=x + a_F \cdot y$.
Combining all equations, we obtain that
\[
    F_{e_i}(G) = F_{e_i}(G') = \frac{F_{e_i}(G_i)}{k} = \frac{x + a_F \cdot y}{k}.
\]
\end{proof}

We are now ready for the final lemma of this section.

\begin{apptheorem}{Lemma}{\ref{lemma:any_graph}}
If edge centrality measure $F$ satisfies Node Deletion, Edge Deletion, Edge Multiplication, Edge Swap, Node Redirect, and Baseline, then $F$ is Edge PageRank.
\end{apptheorem}
\begin{proof}[Proof (Sketch)]
We will prove that $F$ satisfies Edge PageRank recursive equation (Equation~\eqref{eq:pagerank}) with decay factor $a_F$ for every graph $G = (V, E,\phi, b)$  and edge $e \e (v,w) \in E$.
Since this equation uniquely define Edge PageRank, it will imply that $F$ is indeed Edge PageRank.

First assume that $v$ does not have incoming edges.
Then, let us consider graph $G^*$ that is $G$ in which we remove all edges not incident to $v$ and the nodes that become isolated in doing so.
Formally, let $G^*=(V^*,E^+_v(G),\phi|_{E_v^+(G)},b_{V^*})$, where
$V^* = \{\phi_2(e') : e' \in E^+_v(G)\} \cup \{v\}$.
From Edge Deletion and Node Deletion, we get that
\(
    F_e(G) = F_e(G^*).
\)

Observe that the obtained graph is a graph from Lemma~\ref{lemma:star} except sinks might have non-zero node weights.
Thus, let us consider graph $G^\star$ in which we set the weight of all sinks to zero, 
i.e., $G^\star=(V^*,E^+_v(G),\phi|_{E_v^+(G)},b^\star)$, where we have $b^\star(v)=b(v)$ and $b^\star(u)=0$ for every $u \in V' \setminus \{v\}$.
From Sink Weight (Proposition~\ref{prop:loc-sw-se}b), 
\(
    F_e(G^*) = F_e(G^\star),
\)
which means that
\(
    F_e(G) = F_e(G^\star).
\)
On the other hand, from Lemma~\ref{lemma:star} we obtain that $F_e(G'')=b(v)/\deg^+_v(G'')$.
Thus, since $\deg^+_v(G'') = \deg^+_v(G)$, we get
\[
    F_e(G) = \frac{1}{\deg^+_v(G)} \cdot b(v),
\]
which is Equation~\eqref{eq:pagerank} for centrality $F(G)$ and edge $e$.

In the remainder of the proof, let us assume that $v$ has $m>0$ incoming edges, i.e., $E^-_v(G) = \{e_1,\dots,e_m$\}. 
We denote their centralities by $x_i = F_{e_i}(G)$ for every $i \in \{1,\dots,m\}$.
In what follows, through several operations we transform graph $G$ into graph from Lemma~\ref{lemma:star_edge} such that the centrality of $e$ is unchanged.

To this end, for every $i \in \{1,\dots,m\}$, we add to $G$ a simple one-edge graph $G_i=( \{u_i,v_i\},\{e'_i \e (u_i,v_i)\},[x_i,0])$
such that $G_1,\dots,G_m$ are pairwise disjoint and disjoint from G
(see Figure~\ref{fig:app:any_graph}).
Formally, let $G^\circ=(V^\circ,E^\circ,\phi^\circ,b^\circ) = G + G_1 + \dots + G_m$.
Observe that from Locality (Proposition~\ref{prop:loc-sw-se}a) we have that 
$F_{e_i}(G)=x_i$ for every $i \in \{1,\dots,m\}$.
Moreover, for every $i \in \{1,\dots,m\}$ we have $F_{e'_i}(G)=x_i$ from Baseline.
Thus, we obtain that
\begin{equation}
    \label{eq:lemma:any_graph:1}
    F_{e_i}(G^\circ)= x_i =F_{e'_i}(G^\circ) \quad \mbox{for every } i \in \{1,\dots,m\}.
\end{equation}
Observe that from Locality (Proposition~\ref{prop:loc-sw-se}a), we obtain also that
\(
    F_e(G)=F_e(G^\circ).
\)

\begin{figure*}[t]
\centering
\begin{tikzpicture}[x=3cm,y=3cm] 
  \tikzset{     
    e4c node/.style={circle,draw,minimum size=0.54cm,inner sep=0}, 
    e4c edge/.style={above,font=\footnotesize}
  }
  \def\dist{2}
  \def\d{0}
  \def\shift{0.1}
  \def\x{0.35}
  \def\xx{0.5}
  \def\y{0.18}
  \def\yy{0.38}
  \node (G) at (\d - 0.1, \y + 0.15) {$G$};
  \node (G1) at (\d + \shift + \x - 0.18, \yy + 0.1) {$G_1$};
  \node (G2) at (\d + \shift + \x - 0.18, -\yy + 0.1) {$G_2$};
  \node[e4c node] (u_1) at (\d, \y) {}; 
  \node[e4c node] (u_2) at (\d, -\y) {};
  \node[e4c node] (v1) at (\d + \shift + \x, \yy) {$v_1$};
  \node[e4c node] (v2) at (\d + \shift + \x, -\yy) {$v_2$};
  \node[e4c node] (v) at (\d + \shift + \xx, 0) {$v$};
  \node[e4c node] (u1) at (\d + \shift + 2*\x, \yy) {$u_1$};
  \node[e4c node] (u2) at (\d + \shift + 2*\x, -\yy) {$u_2$};
  \node[e4c node] (w_1) at (\d + 2*\xx, \y) {$w$};
  \node[e4c node] (w_2) at (\d + 2*\xx, -\y) {};
  
  \path[->,draw,thick, e4c edge]
  (u_1) edge node[above, pos=0.4] {$e_1$} (v)
  (u_1) edge (u_2)
  (u_2) edge node[above, pos=0.4] {$e_2$} (v)
  (v) edge node[above, pos=0.4] {$e$}  (w_1)
  (v) edge (w_2)
  (w_2) edge (w_1)
  (u1) edge node[above] {$e'_1$} (v1)
  (u2) edge node[above] {$e'_2$} (v2)
  ;

  \def\d{1*\dist}
  \node (G) at (\d -0.05, \yy + 0.1) {$G'$};
  \node[e4c node] (u_1) at (\d, \y) {}; 
  \node[e4c node] (u_2) at (\d, -\y) {};
  \node[e4c node] (v1) at (\d + \shift + \x, \yy) {$v_1$};
  \node[e4c node] (v2) at (\d + \shift + \x, -\yy) {$v_2$};
  \node[e4c node] (v) at (\d + \shift + \xx, 0) {$v$};
  \node[e4c node] (u1) at (\d + \shift + 2*\x, \yy) {$u_1$};
  \node[e4c node] (u2) at (\d + \shift + 2*\x, -\yy) {$u_2$};
  \node[e4c node] (w_1) at (\d + 2*\xx, \y) {$w$};
  \node[e4c node] (w_2) at (\d + 2*\xx, -\y) {};
  
  \path[->,draw,thick, e4c edge]
  (u_1) edge node[above, pos=0.4] {$e_1$} (v1)
  (u_1) edge (u_2)
  (u_2) edge node[above, pos=0.4] {$e_2$} (v2)
  (v) edge node[above, pos=0.4] {$e$}  (w_1)
  (v) edge (w_2)
  (w_2) edge (w_1)
  (u1) edge node[left, pos=0.3] {$e'_1$} (v)
  (u2) edge node[left, pos=0.3] {$e'_2$} (v)
  ;

\end{tikzpicture}
\caption{The scheme of the first part of the proof of Lemma~\ref{lemma:any_graph} for an example graph, $G$, with $m=2$.}
\label{fig:app:any_graph}
\end{figure*}
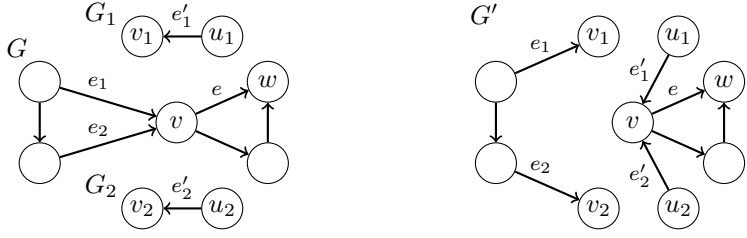

From Equation~\eqref{eq:lemma:any_graph:1}, we see that for every pair of edges $e_i$ and $e'_i$ we can exchange their ends and by Edge Swap we won't affect the centrality of any edge.
Exchanging ends sequentially for all such pairs, we obtain a graph, $G'$, in which incoming edges of $v$, i.e., $E^-_v(G')=\{e'_1,\dots,e'_m\}$, are all source edges, but have the original centralities $x_1,\dots,x_m$ respectively.
Formally, consider
$G' = (V^\circ,E^\circ,\phi',b^\circ)$, where
$\phi'(e_i) = (\phi_1(e_i),v_i)$ and $\phi'(e'_i) = (u_i,v)$ for every $i \in \{1,\dots,m\}$ while $\phi'(e')=\phi^\circ(e')$ for every $e' \in E \setminus E^-_v(G)$.
Then, from Edge Swap, we get that
\(
    F_e(G^\circ)=F_e(G').
\)

Now, observe that in graph $G'$, nodes $u_1,\dots,u_m$ are out-twins (each has one outgoing edge to $v$).
Thus, from Node Redirect, we can sequentially redirect nodes $u_2,\dots,u_m$ into node $u_1$ without affecting the centrality of edge $e$.
Formally, let us denote the graph that we obtain in result by 
$G'' = (V'',E'',\phi'',b'') = R_{u_2 \rightarrow u_1}(\dots(R_{u_m \rightarrow u_1}(G'))\dots)$.
From Node Redirect, we get that
\(
    F_e(G') = F_e(G'').
\)

Next, observe that in graph $G''$ the weight of node $u_1$ is equal to $b''(u_1) = \sum_{i=1}^m b'(u_i) = \sum_{i=1}^m x_i$ (since we redirected all nodes $u_2,\dots,u_m$ into it).
Also, node $v$ has one incoming edge, $e'_1$, that is a source edge.
Thus, by Edge Deletion, removing any edge that is not $e'_1$ or an outgoing edge of $v$ does not affect the centrality of $e$.
Hence, let us remove all edges that are not incident with $v$ and nodes that become isolated in doing so.
Formally, $G^* = (V^*, E^*,\phi''|_{E^*},b''|_{V^*})$, where
$V^* = \{\phi_2(e') : e' \in E^+_v(G)\} \cup \{u_1,v\}$ and 
$E^* = E^+_v(G) \cup \{e'_1\}$.
From Edge Deletion and Node Deletion we get
\(
    F_e(G'') = F_e(G^*).
\)

Observe that the obtained graph is a graph from Lemma~\ref{lemma:star_edge} excepts sinks might have non-zero weights.
Thus, consider $G^\star$ that is $G^*$ in which we set the weights of all sinks to zero, i.e., let
$G^\star = (V^*, E^*,\phi''_{E^*},b^\star)$, where
$b^\star(u_1)=\sum_{i=1}^m x_i$, $b^\star(v) = b(v)$, and $b^\star(u)=0$ for every $u \in V^* \setminus \{u_1,v)$.
From Sink Weight (Proposition~\ref{prop:loc-sw-se}b), we $F_e(G^*) = F_e(G^\star)$, which means that
$F_e(G)=F_e(G^\star)$.
On the other hand, from Lemma~\ref{lemma:star_edge} we get that
\begin{equation}
    \label{eq:lemma:any_graph:2}
    F_e(G^\star) = \frac{1}{\deg^+_v(G^\star)} \left( a_F \cdot \sum_{i=1}^m x_i + b(v) \right). 
\end{equation}

Observe that out-degree of $v$ was not affected in any of our operations, thus $\deg^+_v  (G^\star)  =  \deg^+_v (G)$.
Therefore, since $\sum_{i=1}^m x_i =  \sum_{e_i \in E^-_v(G)}  F_{e_i}  (G)$ and $F_e(G^\star)=F_e(G)$,
Equation~\eqref{eq:lemma:any_graph:2} is equivalent to
Equation~\eqref{eq:pagerank} for centrality $F(G)$ and edge $e$.
This concludes the proof.
\end{proof}

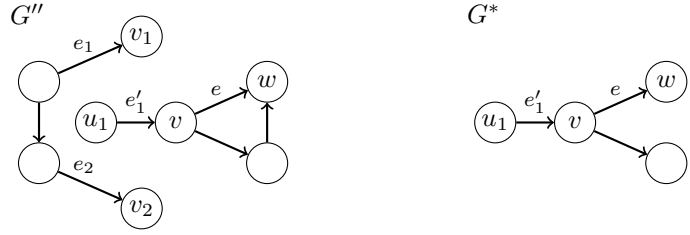
\begin{figure*}[t]
\centering
\begin{tikzpicture}[x=3cm,y=3cm] 
  \tikzset{     
    e4c node/.style={circle,draw,minimum size=0.54cm,inner sep=0}, 
    e4c edge/.style={above,font=\footnotesize}
  }
  \def\dist{2}
  \def\d{0}
  \def\shift{0.1}
  \def\x{0.35}
  \def\xx{0.5}
  \def\y{0.18}
  \def\yy{0.38}

  \def\d{2*\dist}
  \node (G) at (\d -0.05, \yy + 0.1) {$G''$};
  \node[e4c node] (u_1) at (\d, \y) {}; 
  \node[e4c node] (u_2) at (\d, -\y) {};
  \node[e4c node] (v1) at (\d + \shift + \x, \yy) {$v_1$};
  \node[e4c node] (v2) at (\d + \shift + \x, -\yy) {$v_2$};
  \node[e4c node] (v) at (\d + \shift + \xx, 0) {$v$};
  \node[e4c node] (u1) at (\d + 0.5*\xx, 0) {$u_1$};
  \node[e4c node] (w_1) at (\d + 2*\xx, \y) {$w$};
  \node[e4c node] (w_2) at (\d + 2*\xx, -\y) {};
  
  \path[->,draw,thick, e4c edge]
  (u_1) edge node[above, pos=0.4] {$e_1$} (v1)
  (u_1) edge (u_2)
  (u_2) edge node[above, pos=0.4] {$e_2$} (v2)
  (v) edge node[above, pos=0.4] {$e$}  (w_1)
  (v) edge (w_2)
  (w_2) edge (w_1)
  (u1) edge node[above] {$e'_1$} (v)
  ;

  \def\d{3*\dist - 0.5*\xx}
  \node (G) at (\d -0.05 + 0.5*\xx, \yy + 0.1) {$G^*$};
  \node[e4c node] (v) at (\d + \shift + \xx, 0) {$v$};
  \node[e4c node] (u1) at (\d + 0.5*\xx, 0) {$u_1$};
  \node[e4c node] (w_1) at (\d + 2*\xx, \y) {$w$};
  \node[e4c node] (w_2) at (\d + 2*\xx, -\y) {};
  
  \path[->,draw,thick, e4c edge]
  (v) edge node[above, pos=0.4] {$e$}  (w_1)
  (v) edge (w_2)
  %
  (u1) edge node[above] {$e'_1$} (v)
  ;
  
\end{tikzpicture}
\caption{The scheme of the second part of the proof of Lemma~\ref{lemma:any_graph} for $G$ from Figure~\ref{fig:app:any_graph}}
\label{fig:app:any_graph:2}
\end{figure*}

\end{document}